%% file: main.tex
\documentclass{article}

\usepackage[affil-it]{authblk}
\usepackage[usenames,dvipsnames]{xcolor}
\usepackage{amsfonts}
\usepackage{amsmath,amsthm,amssymb,dsfont}
\usepackage{enumerate}
\usepackage{graphicx}	
\usepackage{subcaption}
\usepackage[margin=3cm]{geometry}
\usepackage{url}
\usepackage{todonotes}
\usepackage{bbm}

\usepackage{tikz}

\usepackage{pifont}
\usepackage{multirow}
\usepackage{makecell}

\usepackage{epsfig}
\usetikzlibrary{shapes.symbols,patterns} 
\usepackage{pgfplots}
\pgfplotsset{compat=1.10}
\usepgfplotslibrary{fillbetween}

\definecolor{linkblue}{HTML}{001487}
\usepackage{hyperref}
\hypersetup{colorlinks=true,citecolor=linkblue,linkcolor=linkblue,filecolor=linkblue,urlcolor=linkblue,breaklinks=true}

\usepackage{nicefrac}
\usepackage{mathtools}
\usepackage[nameinlink,capitalize,noabbrev]{cleveref}

\usepackage{algorithm}
\usepackage{algorithmic}

\usepackage{mdframed}
\usepackage{aligned-overset}

\usepackage{dsfont}
\usepackage{qcircuit}

\theoremstyle{plain}
\newtheorem{theorem}{Theorem}[section]
\newtheorem{lemma}[theorem]{Lemma}

\newtheorem{corollary}[theorem]{Corollary}
\newtheorem{proposition}[theorem]{Proposition}

\crefname{question}{Question}{Questions}

\theoremstyle{definition}
\newtheorem{definition}[theorem]{Definition}

\newtheorem{example}[theorem]{Example}

\newcommand*{\cD}{\mathcal{D}}
\newcommand*{\cE}{\mathcal{E}}
\newcommand*{\cF}{\mathcal{F}}
\newcommand*{\cG}{\mathcal{G}}

\newcommand*{\cJ}{\mathcal{J}}
\newcommand*{\cL}{\mathcal{L}}
\newcommand*{\cI}{\mathcal{I}}

\newcommand*{\cN}{\mathcal{N}}
\newcommand*{\cM}{\mathcal{M}}
\newcommand*{\cP}{\mathcal{P}}

\newcommand*{\cU}{\mathcal{U}}

\newcommand*{\LO}{\mathrm{LO}}
\newcommand*{\LOCC}{\mathrm{LOCC}}
\newcommand*{\PPT}{\mathrm{PPT}}
\newcommand*{\SEP}{\mathrm{SEP}}

\newcommand*{\St}{\mathrm{S}}

\newcommand*{\sgn}{\mathrm{sgn}}

\newcommand*{\cptp}{\mathrm{CPTP}}

\newcommand*{\id}{\mathds{1}}

\newcommand*{\tr}{\mathrm{tr}}
\newcommand*{\ket}[1]{| #1 \rangle}
\newcommand*{\bra}[1]{\langle #1 |}

\newcommand{\proj}[1]{|#1\rangle\langle #1|}
\newcommand*{\braket}[1]{\langle #1 \rangle}
\newcommand*{\abs}[1]{| #1 |}

\newcommand*{\qi}{\mathrm{QI}}

\newcommand*{\ops}{\mathrm{L}}
\newcommand*{\supops}{\mathrm{T}}
\newcommand*{\states}{\mathrm{D}}

\newcommand*{\conv}{\mathrm{conv}}
\newcommand*{\CPTP}{\mathrm{CPTP}}
\newcommand*{\CP}{\mathrm{CP}}
\newcommand*{\TP}{\mathrm{TP}}
\newcommand*{\CPTN}{\mathrm{CPTN}}
\newcommand*{\HP}{\mathrm{HP}}
\newcommand*{\HPTP}{\mathrm{HPTP}}

\newcommand{\norm}[1]{\left\lVert#1\right\rVert}

\DeclareRobustCommand{\abbrevcrefs}{%
\Crefname{theorem}{Thm.}{Thms.}%
\Crefname{corollary}{Cor.}{Cors.}%
\Crefname{lemma}{Lem.}{Lems.}%
\Crefname{proposition}{Prop.}{Props.}%
\Crefname{equation}{Eq.}{Eqs.}%
\Crefname{example}{Ex.}{Exs.}%
}

\DeclareRobustCommand{\Cshref}[1]{{\abbrevcrefs\Cref{#1}}}

\usepackage{arydshln}
\usepackage{bbold}
\usepackage{mathabx}

\usepackage{float}

 \allowdisplaybreaks

\title{Circuit cutting with classical side information}

 \author{\normalsize Christophe Piveteau$^{*\, 1}$, Lukas Schmitt$^{*\, 1,2}$ and David Sutter$^{2}$}
  \affil{\small $^{1}$Institute for Theoretical Physics, ETH Zurich\\
  $^{2}$IBM Quantum, IBM Research Europe -- Zurich
 }
 \date{}

\begin{document}

\maketitle
\def\thefootnote{*}\footnotetext{These authors equally contributed to this work.}
\def\thefootnote{\arabic{footnote}}

\begin{abstract}
Circuit cutting is a technique for simulating large quantum circuits by partitioning them into smaller subcircuits, which can be executed on smaller quantum devices. The results from these subcircuits are then combined in classical post-processing to accurately reconstruct the expectation value of the original circuit. Circuit cutting introduces a sampling overhead that grows exponentially with the number of gates and qubit wires that are cut.
Many recently developed quasiprobabilistic circuit cutting techniques leverage classical side information, obtained from intermediate measurements within the subcircuits, to enhance the post-processing step. In this work, we provide a formalization of general circuit cutting techniques utilizing side information through quantum instruments. With this framework, we analyze the advantage that classical side information provides in reducing the sampling overhead of circuit cutting.
Surprisingly, we find that in certain scenarios, side information does not yield any reduction in sampling overhead, whereas in others it is essential for circuit cutting to be feasible at all. Furthermore, we present a lower bound for the optimal sampling overhead with side information that can be evaluated efficiently via semidefinite programming and improves on all previously known lower bounds.

\end{abstract}


\section{Introduction}
A major difficulty in the realization of useful quantum computation is the limited number of available qubits.
This problem is especially pressing, considering that quantum error correction is likely to require a large number of physical qubits to realize even a few logical qubits.
One approach to solving this issue is to connect multiple quantum devices using quantum communication links over a network.
However, such hardware is likely to remain out of reach in the coming years.

To address this shortcoming, recent years have seen the advent of circuit cutting techniques~\cite{BSS16,PHOW20,Mitarai_2021,MF_21,piv23,SPS24,anguspaper,almu_nature_04}.  
The idea is to partition large quantum circuits into smaller subcircuits that can each be executed on a smaller quantum device.
Using appropriate classical post-processing, the expectation value of an observable in the original circuit can be reconstructed with arbitrary precision --- at the cost of a larger variance. This implies that to achieve a given (fixed) accuracy we require more executions of the circuit. This \emph{sampling overhead} scales exponentially in the number of gates and qubit wires that were ``cut'' in the circuit partitioning.
Despite its inevitable exponential overhead, circuit cutting is of high practical importance, and hence it is of interest to minimize the overhead as much as possible.

In this paper, we focus on the most popular and widely studied form of circuit cutting, which is based on the quasiprobability simulation technique.
We note that other approaches do exist~\cite{BSS16,forging22}, though in comparison, they typically lack the wide applicability.
The key idea in quasiprobabilistic circuit cutting is to decompose the entangling gates across the circuit partitions into operations that act locally on each circuit partition.
More mathematically, for a bipartite unitary channel $\cU_{AB}$, one studies \emph{quasiprobability decompositions} (QPD)
\begin{align}
    \cU_{AB} = \sum_{i=1}^m a_i \cF_i \, ,
\end{align}
where $m\geq 1$ and $a\in\mathbb{R}^m$.
Depending on whether we allow the circuit partitions to exchange classical communication or not, $\cF_i$ are either local operations ($\LO$) or local operations with classical communication ($\LOCC$).
Given such a QPD, the quasiprobability simulation technique can simulate the nonlocal unitary $\cU$ by randomly replacing it by $\cF_i$ and performing appropriate classical post-processing of the circuit output.
This leads to a sampling overhead of $\norm{a}_1^2$ where $\norm{a}_1\coloneqq \sum_i\abs{a_i}$.
When the QPDs of multiple nonlocal gates in the circuit are combined together, the $\ell_1$-norm of their coefficients scales multiplicatively, which leads to the exponential scaling of the sampling overhead in circuit cutting.
For example, a circuit with $n$ nonlocal gates $\cU$ to be cut will incur a sampling overhead of $\norm{a}_1^{2n}$.

Clearly, it is desirable to pick a QPD that leads to the smallest possible sampling overhead.
This optimal overhead is called the \emph{quasiprobability extent}\footnote{Previous work sometimes called this quantity ``$\gamma$-factor''~\cite{piv23,SPS24,BPS23}. We instead adopt the naming from~\cite{christophe_phd}.} $\gamma(\cE)$ of a bipartite channel $\cE$, which is defined as the smallest $\ell_1$-norm $\norm{a}_1$ over all valid QPDs of $\cE$.
To differentiate between situations in which the $\cF_i$ are allowed to use classical communication or not, we denote the $\LO$-assisted and $\LOCC$-assisted extent by $\gamma_{\LO}$ and $\gamma_{\LOCC}$, respectively.
Clearly, as $\LO \subseteq \LOCC$, one has $\gamma_{\LOCC}\leq\gamma_{\LO}$.

The quasiprobability simulation method can be generalized by allowing the operations $\cF_i$ to additionally produce classical information e.g.~by performing intermediate measurements, which is then incorporated into the classical post-processing procedure.
We refer to this improved technique as quasiprobability simulation with \emph{classical side information}.
The utilization of classical side information can be mathematically captured by allowing certain non-completely positive and non-trace-preserving maps to be used in the QPD.
This enlarged sets of allowed operations will be denoted by $\LO^{\star}$ and $\LOCC^{\star}$, and the corresponding extents by $\gamma_{\LO^{\star}}$ and $\gamma_{\LOCC^{\star}}$ respectively.
Since $\LO \subseteq \LO^\star$ and $\LOCC \subseteq \LOCC^\star$ we have $\gamma_{\LO} \geq \gamma_{\LO^\star}$ and $\gamma_{\LOCC} \geq \gamma_{\LOCC^\star}$.
Given that there is not much reason not to allow for classical side information in practice, and we want to minimize the sampling overhead, $\gamma_{\LO^\star}$ and $\gamma_{\LOCC^\star}$ are the operationally relevant quantities for practical applications.

While specific quasiprobabilistic protocols utilizing classical side information have been around in literature for a while (see for example~\cite{endo18} in the context of quantum error mitigation and~\cite{PHOW20,Mitarai_2021} for quasiprobabilistic circuit cutting), only recent work~\cite{piv23,SPS24} formalized the optimal sampling overheads with side information through the quantities $\gamma_{\LO^{\star}}$ and $\gamma_{\LOCC^{\star}}$.\footnote{In~\cite{piv23,SPS24} $\gamma_{\LO^{\star}}$ and $\gamma_{\LOCC^{\star}}$ were denoted by $\gamma_{\LO}$ and $\gamma_{\LOCC}$.}
Later, this formalization has been rephrased in different terms in~\cite{anguspaper}.
However, the importance of classical side information for the purpose of circuit cutting has remained poorly understood.
For instance, previous work does not address whether classical side information is necessary in the first place and whether it enables a smaller sampling overhead.

The quantities $\gamma_{\LO^\star}$ and $\gamma_{\LOCC^\star}$ have proven to be challenging to compute.
This is not surprising, as they are closely related to entanglement measures~\cite{piv23}, which are known to be NP-hard to evaluate in general~\cite{sevag2010_nphard}.
We know the exact values of these quantities only for a few special cases, such as for unitaries with a Schmidt-like decomposition~\cite{SPS24,anguspaper} and for pure states~\cite{piv23,vidal99}.
For more general quantum channels, only little is known about their extent.
Recent work~\cite{xin24} has made progress in this direction, by finding lower bounds of $\gamma_{\LOCC}$ through a relaxation of $\LOCC$ to the set of positive partial transpose channels ($\PPT$), as $\gamma_{\LOCC}\geq\gamma_{\PPT}$ where $\gamma_{\PPT}$ is defined analogously to $\gamma_{\LO}$ and $\gamma_{\LOCC}$.
The quantity $\gamma_{\PPT}$ is expressible in terms of a semidefinite program (SDP) and can thus be efficiently evaluated numerically. 
Unfortunately, these $\PPT$-based lower bounds are of limited applicability, as they only bound $\gamma_{\LOCC}$, but not $\gamma_{\LOCC^\star}$, which is the physically relevant quantity.
In other words, they do not make any statement about the performance of circuit cutting assisted by classical side information.

In this paper, we advance the understanding of circuit cutting with classical side information with following contributions:
\begin{enumerate}[(i)]
    \item We mathematically formalize quasiprobabilitsic circuit cutting with classical side information using the notion of \emph{quantum instruments}.
    Essentially, we replace the channels $\cF_i$ with quantum instruments, which describe the most general quantum processes that output classical information.\footnote{Our work generalizes the formalism for classical side information in the quasiprobability simulation of non-physical operations, which was recently introduced in~\cite{zhao2023power}.}
    \item We show that classical side information is crucial for circuit cutting in the absence of classical communication.
    We prove that $\gamma_{\LO}(\cE)$ is finite if and only if $\cE$ is non-signaling, whereas $\gamma_{\LO^\star}(\cE)$ is finite for all $\CPTP$ maps $\cE$. We refer to~\cref{prop_LO_vs_LO_star} for the precise statement.
    \item We show that classical side information is useless in the $\PPT$ setting, i.e.~$\gamma_{\PPT}(\cE)=\gamma_{\PPT^\star}(\cE)$ for all $\CPTP$ maps $\cE$ (see~\cref{prop_PPT_PTT_star}).
    This implies that the SDP lower bound introduced in~\cite{xin24} is also applicable to $\gamma_{\LOCC^\star}$ and not just $\gamma_{\LOCC}$ as discussed in~\cref{sec_ppt}.
    \item We study the $\PPT$-based lower bound and show that it is tight in certain situations and always better than previously known lower bounds. We refer to~\cref{sec_ppt} and in particular to~\cref{corr_statebound} for the rigorous statements.
\end{enumerate}
The second and third result are summarized in~\cref{fig_results}.
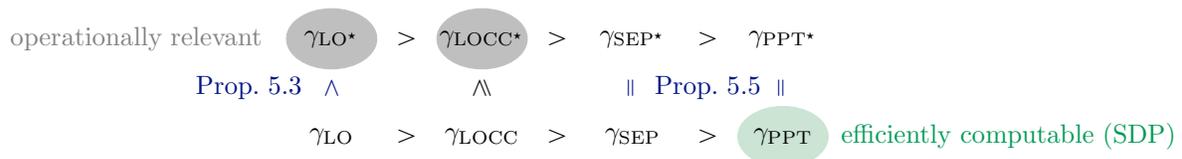
\begin{figure}[!htb]
    \centering
    \input{results_fig}
\caption{Relations between different sampling overheads for $\CPTP$ maps. A strict inequality $>$ implies that there exists a $\CPTP$ map for which both extents differ. The gray quantities $\gamma_{\mathrm{LO}^\star}$ and $\gamma_{\mathrm{LOCC}^\star}$ are the physically relevant quantities that capture the optimal overhead of circuit cutting without and with classical communication. The green quantity $\gamma_{\mathrm{PPT}}$ can be computed efficiently via a semidefinte program (SDP) and is a lower bound to the physically relevant quantities. Here, $\SEP$ denotes the set of all separable quantum channels. The sets $\mathrm{LO}^\star, \mathrm{LOCC}^\star, \mathrm{SEP}^\star, \mathrm{PPT}^\star$ are extensions of $\mathrm{LO}, \mathrm{LOCC}, \mathrm{SEP}, \mathrm{PPT}$ where we additionally allow for classical side information to be incorporated in the classical post-processing.  We refer to~\cref{sec_preliminaries} for a precise definition of these sets. In~\cref{app:remaining_inequalities} we show that all horizontal inequalities are strict. }
\label{fig_results}
\end{figure}

\section{Preliminaries}\label{sec_preliminaries}
\subsection{Notation}
For two finite-dimensional complex Hilbert spaces $I$, $O$, we denote by $\ops(I\rightarrow O)$ the set of linear operators from $I$ to $O$ and by $\supops(I \rightarrow O)\coloneqq \ops(\ops(I \rightarrow I),\ops(O \rightarrow O))$ the set of linear superoperators from $I$ to $O$.
We denote the set of completely positive maps by $\CP(I \rightarrow O)$ and the set of trace-preserving maps by $\TP(I \rightarrow O)$.
The set of quantum channels is $\CPTP(I \rightarrow O)\coloneqq\CP(I \rightarrow O)\cap\TP(I \rightarrow O)$.
The set of completely positive trace non-increasing maps is denoted by $\CPTN(I \rightarrow O)$ and the set of Hermitian preserving superoperators by $\HP(I\rightarrow O)$.
The set of trace-preserving, Hermitian preserving maps is denoted $\HPTP(I\rightarrow O)$.
When the input and output systems are the same, we drop the arrow in the notation, for example $\ops(I \rightarrow I)=\ops(I)$ and $\supops(I \rightarrow I)=\supops(I)$.
The set of density matrices on $O$ is denoted by $\states(O)$. For pure (rank one) states we use the symbol $\overline{\mathrm{D}}(O)$. 
The Choi isomorphism $\cJ:\supops(I\rightarrow O)\rightarrow\ops(I'O)$ is defined as $\cJ(\cE)\coloneqq (\cI_{I'}\otimes\cE)(\proj{\Psi_{I'I}})$ where $I'$ is a copy of $I$ and $\ket{\Psi_{I'I}}\coloneqq (\dim(I))^{-1/2} \sum_i \ket{i}_{I'}\ket{i}_{I}$ is a maximally entangled state across $I'$ and $I$ for some fixed orthonormal basis $\ket{i}$.
It is well known that for a linear superoperator $\cE\in\supops(I \rightarrow O)$ one has $\cE\in\CP(I \rightarrow O)\Leftrightarrow\cJ(\cE)\geq 0$ and $\cE\in\TP(I \rightarrow O)\Leftrightarrow\tr_O[\cJ(\cE)]=\frac{1}{\dim(I)}\id_I$.

The most general measurement process in quantum mechanics can be described in terms of a channel $\cE_{I\rightarrow OX}\in\cptp(I\rightarrow OX)$ which maps every state to a classical-quantum state
\begin{equation}\label{eq_instrument}
    \cE_{I\rightarrow OX}(\rho_I) = \sum_i \cE_i(\rho_I)\otimes \proj{i}_X \, ,
\end{equation}
for some $\cE_i\in\supops(I \rightarrow O)$, where $I$ and $O$ are the input and output quantum systems and $X$ is a classical register that stores the measurement outcome.
It is easy to check that all $\cE_i$ must lie in $\CPTN(I\rightarrow O)$ and, additionally, $\sum_i\cE_i = \tr_X\circ \cE_{I\rightarrow OX} \in\CPTP(I\rightarrow O)$.
In contrast, all sets of operations $\cE_i$ that meet these two properties describe a quantum measurement by~\Cref{eq_instrument}.
This motivates the definition of a (discrete) quantum instrument from $I$ to $O$ as a discrete family $(\cE_i)_i$ of nonzero completely positive operators $\cE_i\in\CP(I \rightarrow O)$ such that $\sum_i\cE_i\in\CPTP(I \rightarrow O)$.
The set of discrete quantum instruments from $I$ to $O$ is denoted $\qi(I \rightarrow O)$.

We say a $(\cF_j)_{j\in J}\in\qi(I\rightarrow O)$ is a coarse graining of $(\cE_i)_{i\in K}\in\qi(I\rightarrow O)$ if there exists a surjective map $f:K\rightarrow J$ such that \smash{$\forall j\in J: \cF_j = \sum_{i:f(i)=j} \cE_i$}.
The intuitive meaning is that $(\cF_j)_j$ can be understood as the measurement $(\cE_i)_i$ where some of the measurement outcomes are treated as the same.
So, the instrument $(\cF_j)_j$ can be realized by first implementing $(\cE_i)_i$ and then discarding some classical information.

A bipartite state $\rho\in\states(A B)$ is not entangled (also called separable) if it can be written as
\begin{equation}
    \rho =\sum_i p_i \sigma_i\otimes \tau_i \, ,
\end{equation}
for some probability distribution $(p_i)_i$ and $\sigma_i\in\states(A),\tau_i\in\states(B)$.
The set of separable states is denoted by $\SEP(A,B)$.
Testing for separability is challenging, as it is not a convex constraint.
For this reason, a common technique is to relax this constraint by defining the larger set
\begin{equation}
    \PPT(A,B) \coloneqq \{\rho\in\states(AB) : \rho^{T_B}\geq 0\} \, ,
\end{equation}
where $\rho^{T_B}$ denotes the partial transposition of $\rho$ on the system $B$.
Indeed, $\SEP(A,B)\subset\PPT(A,B)$ as any separable state has a non-negative partial transpose.\footnote{Later, we will also denote the set of separable and $\PPT$ chanels with the symbols $\SEP$ and $\PPT$. It will always be clear from context which set is considered.}
This is the well-known $\PPT$ criterion which states that any state with non-positive partial transpose must be entangled~\cite{peres96,horodecki96}.
\subsection{Quasiprobability simulation}\label{sec:qps_technique}
Suppose we are given a quantum circuit that prepares some state and then measures an observable.
We are interested in predicting the expectation value of this observable with a quantum computer that is restricted in the sense that it cannot realize a full universal gate set.
Let us separate the set of all quantum operations into the ``free'' operations (denoted by the set $S$) which the quantum computer can physically execute, as well as ``non-free'' operations\footnote{The expression ``free'' is inspired by the literature of quantum resource theories. One can often equate the decomposition set $S$ with the free operations of some quantum resource theory -- for more details, see \cite{christophe_phd}.}.
If the desired quantum circuits contains non-free operations, it is a priori impossible to directly execute the circuit on the given quantum computer.
However, using the quasiprobability simulation technique, it is still possible to evaluate the expectation value of the original circuit with the restricted quantum computer, albeit with a sampling overhead.

Quasiprobability simulation has been applied to a variety of different applications, each corresponding to a different choice of the set of free operations $S$.
For example, it can be used to simulate near-Clifford circuits on a classical computer~\cite{PWB15,HC17} (here $S$ is chosen to be Clifford operations) and for quantum error mitigation~\cite{TBG17,endo18} (here $S$ is chosen to be the set of noisy operations executable by the quantum computer).
In the context of quasiprobabilistic circuit cutting, the set $S$ is chosen to contain local operations across the circuit bipartitions, and quasiprobability simulation allows for the simulation of non-local gates~\cite{Mitarai_2021,MF_21,piv23}.
We will elaborate more on the application to circuit cutting in \Cref{sec:circuit_cutting}, and in the following, we briefly summarize how the quaspirobability technique operates.
For a comprehensive treatment, we refer the reader to~\cite{christophe_phd}.
Suppose that our task is to simulate some non-free operation $\cE\in\supops(I\rightarrow O)$ using free operations.
Furthermore, assume that we can decompose $\cE$ as a linear combination using elements $\cF_i\in S\subset\supops(I\rightarrow O)$
\begin{equation} \label{eq_QPD2}
    \cE = \sum_{i=1}^m a_i \cF_i \, ,
\end{equation}
where $a_i\in\mathbb{R}$.
This decomposition is called a \emph{quasiprobability decomposition} (QPD), because some of the coefficients $a_i$ might generally be negative.
Using some simple algebra, we can rewrite this QPD as to remove the negativity form the weight and make it a proper probabilistic mixture
\begin{equation}\label{eq_QPD3}
    \cE = \sum_{i=1}^m p_i \mathcal{F}_i \, \mathrm{sign}(a_i) \norm{a}_1\, ,
\end{equation}
with the probability distribution $p_i = |a_i|/\norm{a}_1$.
The idea of quasiprobability simulation is to probabilistically replace the operation $\cE$ by $\cF_i$ with probability $p_i$, and then weight the final measurement outcome by $\mathrm{sign}(a_i) \norm{a}_1$.
To see that this indeed provides us with an unbiased estimate of the expectation value, let us consider the simplest instance where the circuit consists purely of the operation $\cE$ followed by the measurement of the observable $O$.
Then,
\begin{align}
    \tr[O \cE(\rho)]
  = \sum_{i=1}^m a_i \tr[O \cF_i(\rho)]
  = \sum_{i=1}^m p_i \tr[O \cF_i(\rho)] \sgn(a_i)\norm{a}_1 \, .  \label{eq:expval_monte_carlo}
\end{align}
Hence, the random replacement of $\cE$ by $\cF_i$, coupled with the appropriate post-processing, provides an unbiased estimator for $\tr[O\cE(\rho)]$.
However, since we are scaling up the circuit output by the factor $\pm\norm{a}_1$, this estimator generally has a higher variance compared to directly executing the target circuit in the first place.
Therefore, a higher number shots is needed in order to reach the same confidence for the expectation value.
\begin{proposition}\label{prop_qpd_overhead}
Let $\epsilon > 0$ and $\delta\in (0,1]$.
Using $\frac{2\norm{a}^2_1}{\epsilon^2} \log\left(\frac{2}{\delta}\right)$ shots, quasiprobability simulation produces an estimate for $\tr[O \cE(\rho)]$ up to an additive error $\epsilon$ with probability at least $1-\delta$.
\end{proposition}
This proposition is a simple consequence of Hoeffding's inequality, and we refer the reader to~\cite{PWB15,TBG17,endo18} for more details.
Since the number of shots required increases multiplicatively in $\|a\|_1^2$, it is desirable to pick a QPD that has the smallest possible $\ell_1$-norm of its coefficients.
Given a fixed set $S$ of operations that our computer can perform, we define the extent of a map $\cE$ as the smallest such achievable $\ell_1$-norm.
\begin{definition}\label{def:gamma}
  For a decomposition set $S\subset\supops(I\rightarrow O)$ and a target operation $\cE\in\supops(I \rightarrow O)$, we define the \emph{quasiprobability extent} of $\cE$ w.r.t. $S$ as
  \begin{equation}\label{eq:gamma}
    \gamma_{S}(\cE) \coloneqq \inf \Big\{ \sum_{i=1}^m \lvert a_i \rvert : \cE = \sum\limits_{i=1}^m a_i \cF_i \text{ where } m\geq 1, \cF_i\in S \textnormal{ and } a_i \in \mathbb{R} \Big\} \, .
  \end{equation}
\end{definition}
If no valid decomposition exists, we define the extent to be infinite.
An infinite extent indicates that a quasiprobabilistic simulation of $\cE$ with a decomposition set $S$ is not possible.

Previously, we have illustrated how quasiprobability simulation can estimate the expectation value of an observable on a circuit which contains one single non-free operation $\cE$.
The technique can naturally be applied to circuits which contain multiple such gates, given that each gate exhibits its own QPD into operations in $S$.
The main idea is to simply combine the individual QPDs into a large QPD of the total circuit, such that the previous consideration is again applicable.
As an example, suppose we have three operations $\cE_1,\cE_2,\cE_3$ that do no lie in $S$ and we want to implement the circuit 
\begin{equation} \label{eq_example_setting}
  \cU_4\circ \cE_3 \circ \cU_3 \circ \cE_2 \circ \cU_2 \circ \cE_1 \circ \cU_1 \, ,
\end{equation}
where $\cU_i$ denote concatenations of operations that lie in $S$.
Assuming we have the QPDs $\cE_j = \sum_{i}a^{(j)}_{i}\cF^{(j)}_{i}$, \Cref{eq_example_setting} can be rewritten as
\begin{equation}
  \sum\limits_{i_1,i_2,i_3} p^{(1)}_{i_1}p^{(2)}_{i_2}p^{(3)}_{i_3} \cU_4\circ \cF_{i_3}^{(3)} \circ \cU_4 \circ \cF_{i_2}^{(2)} \circ \cU_2 \circ \cF_{i_1}^{(1)} \circ \cU_1 \, \sgn(a_{i_1}^{(1)}a_{i_2}^{(2)}a_{i_3}^{(3)})\|a^{(1)}\|_1 \|a^{(2)}\|_1 \|a^{(3)}\|_1 
\end{equation}
where $p_i^{(j)}\coloneqq \abs{a_i^{(j)}} / \norm{a^{(j)}}_1$.
The Monte Carlo sampling procedure in this case works as follows: For every instance of an unsupported operation $\cE_j$ in the circuit, we separately and independently sample a random gate $\cF_i^{(j)}$ according to \smash{$p^{(j)}_i$} to replace it.
The sampled indices ($i_1$, $i_2$ and $i_3$) are stored, and at the end of the circuits we weight the measurement outcome by $\pm \|a^{(1)}\|_1 \|a^{(2)}\|_1 \|a^{(3)}\|_1$.
Note how the $\ell_1$-norms of the individual decompositions combine in a multiplicative fashion, so the overall sampling overhead generally scales exponentially in the number of unsupported operations.

The possibly infinite number of elements that a QPD can contain might seem daunting to deal with.
Fortunately, in many cases the decomposition set $S$ is convex, allowing for a simpler characterization with QPDs containing only two elements.
\begin{lemma}\label{lemma_convex}
    Let $S\subset\supops(I \rightarrow O)$ be a convex set and $\cE\in\supops(I \rightarrow O)$. Then
    \begin{equation}
        \gamma_{S}(\cE) = \inf \{a^+ + a^- :\cE=a^+\cE^+ - a^-\cE^- \text{ where } a^{\pm}\geq 0\text{ and } \cE^{\pm}\in S\} \, .
    \end{equation}
\end{lemma}
\begin{proof}
    It suffices to show that any QPD of $\cE$ with respect to $S$ induces another QPD that has at most two elements and the same $\ell_1$-norm of its coefficients.
    Let $\cE=\sum_i a_i\cF_i$ be one such QPD.
    Define
    \begin{equation}
        a^+ \coloneq \sum_{i:a_i\geq 0} a_i,     
        \qquad      a^- \coloneq -\sum_{i:a_i< 0} a_i,
        \qquad \cE^+ \coloneq \frac{1}{a^+}\sum_{i:a_i\geq 0}a_i\cE_i,     
        \quad \textnormal{and}\quad      \cE^- \coloneq \frac{1}{a^-}\sum_{i:a_i< 0}a_i\cE_i \, .
    \end{equation}
    Clearly, $\cE = a^+\cE^+ - a^-\cE^-$ and $a^++a^- = \sum_i\lvert a_i\rvert$ and $\cE^{\pm}\in S$ by convexity.
\end{proof}

The quasiprobability simulation technique cannot only simulate non-free quantum channels, but also non-free states.
Indeed, states can be interpreted as a special case of quantum channels with trivial input spaces, so the above methods can be readily applied.\footnote{In some applications of quasiprobability simulation, the focus more commonly lies on the simulation of non-free states (instead of channels). For instance, the classical simulation of Clifford+T circuits can be achieved by simulating magic states with stabilizer states~\cite{HC17,HG19}.}
Correspondingly, the quasiprobability extent of a state $\rho\in\states(A)\equiv \CPTP(\mathbb{C}\rightarrow A)$ w.r.t. to some decomposition set $S\subset\supops(\mathbb{C}\rightarrow A)$ is given by
\begin{align}
    \gamma_{S}(\rho) \coloneqq \inf \Big\{ \sum_{i=1}^m \lvert a_i \rvert :  \rho = \sum\limits_{i=1}^m a_i \sigma_i \text{ where } m\geq 1, \sigma_i\in S \textnormal{ and } a_i \in \mathbb{R} \Big\} \, .
\end{align}
If the decomposition set $S$ is convex, one can connect the extent to a corresponding robustness measure.
The robustness of a state $\rho$ with respect to a decomposition set $S$ is given by
\begin{align}
    R_{S}(\rho)\coloneqq \min_{\sigma \in S,t\geq 0} \Big\{ t : \frac{\rho+t \sigma}{1+t} \in S \Big \} \, .
\end{align}

\begin{lemma}\label{lem_robustness}
    Let $S$ be a convex set and $\rho$ a quantum state. Then,
    \begin{equation}\label{gamma_robustness}
        \gamma_{S}(\rho) = 1 + 2R_{S}(\rho)  \, .
    \end{equation}
\end{lemma}
\begin{proof}
    Based on \cref{lemma_convex}, it is easy to see that any optimal decomposition of one quantity induces a valid decomposition for the other. We refer to~\cite[Lemma~3.17]{christophe_phd} for a more detailed proof that also holds for channels.
\end{proof}

\section{Quasiprobability simulation with classical side information} \label{sec_intermediate_meas}
It might seem natural to restrict ourselves to considering decomposition sets $S$ that are subsets of the completely positive trace-preserving (CPTP) maps, since these are precisely the physically realizable operations.
In this section, we will see that it is useful to allow $S$ to also contain certain trace non-increasing maps and even non-completely positive maps.
This should be understood as a mathematical trick, that allows us to capture the notion of \emph{classical side information} in the quasiprobability simulation procedure.
As explained in \Cref{sec:qps_technique}, the procedure operates by randomly replacing non-free operations by free ones, and then appropriately multiplying the output of the circuit by some real scalar.
This technique can be augmented by allowing this post-processing step to additionally multiply the output by some scalar which may depend on mid-circuits measurements.
Let us first consider a few explanatory examples before presenting the trick in full generality.

As an example, consider that we want to simulate the one-qubit trace non-increasing operation $\Pi_0(\cdot)=\proj{0}(\cdot)\proj{0}$ which maps $\rho\mapsto \braket{0|\rho|0}\proj{0}$.
This map cannot be written as a linear combination of CPTP maps, because any such linear combination would be proportional to a trace-preserving map.
Hence, $\Pi_0$ cannot be simulated if $S$ only contains CPTP maps (i.e., if no classical side information is utilized).
Yet still, it can be simulated by performing a measurement in the computational basis described by
\begin{equation}
  \mathcal{M}(\rho) \coloneqq \sum_{i=0}^1 \proj{i} \rho \proj{i} \otimes \proj{i}_X
\end{equation}
which stores the outcome in a classical register $X$.
Notice that one can write $\Pi_0(\rho) = (\id\otimes \bra{0}_X) \mathcal{M}(\rho)  (\id\otimes \ket{0}_X)$.
If we now insert this in the term for the expectation value of some observable $O$ on some circuit with initial state $\rho_{\mathrm{in}}$ and operations $\mathcal{U}_2\circ \Pi_0\circ \mathcal{U}_1$ (here $\mathcal{U}_i$ stand for all other operations that occur before and after in the circuit), then the $\bra{0}_X$ and $\ket{0}_X$ can be absorbed into the observable, i.e.
\begin{equation}\label{eq:side_info_example}
  \tr\left[O \cdot \mathcal{U}_2\circ\Pi_0\circ\mathcal{U}_1 (\rho)\right]
  = 
  \tr\left[(O\otimes\proj{0}_X) \mathcal{U}_2\circ\mathcal{M}\circ\mathcal{U}_1 (\rho)\right]
  \, .
\end{equation}
Hence, the expectation value can be estimated as follows: We replace the $\Pi_0$ operation in the circuit by the computational basis measurement $\cM$.
For every circuit execution, we weight the measurement output either by $1$ or $0$, depending on the mid-circuit measurement outcome.
Note that this is not post-selection, as the ``discarded'' runs are still taken into account when averaging over multiple runs, they just have value zero.
\Cref{eq:side_info_example} guarantees that this procedure constitutes an unbiased estimator.

It is illustrative to consider a second example, this time of the non-completely positive map $\cE=\Pi_0-\Pi_1$ where we analogously define $\Pi_1(\cdot)=\proj{1}(\cdot)\proj{1}$.
Here again, if we only care about the final observable, we can simulate this $\cE$ by replacing it with a computational basis measurement
\begin{equation}
  \tr\left[O \cdot \mathcal{U}_2\circ(\Pi_0-\Pi_1)\circ\mathcal{U}_1 (\rho)\right]
  = 
  \tr\left[(O\otimes(\proj{0}_X-\proj{1}_X)) \cdot \mathcal{U}_2\circ\mathcal{M}\circ\mathcal{U}_1 (\rho)\right]
  \, 
\end{equation}
and by weighting the final circuit observable outcome by either $+1$ or $-1$ depending on the intermediate measurement outcome.

The main insights of the previous two examples is that \emph{for the purpose of estimating some expectation value} one can use classical side information to simulate non-CPTP maps by weighting the final observable measurement by some factor depending on the intermediate measurement outcome.
The most general operation that produces such classical information is described by a quantum instrument.
Given a discrete quantum instrument $(\cG_j)_j$ that our computer can physically perform, this technique allows us to effectively run any operation $\sum_j b_j \cG_j$ for some weighting factors $b_j\in\mathbb{R}$. 
This is useful for quasiprobabilty simulation, because when we now decompose our target operation $\cE=\sum_i a_i\cF_i$, the $\cF_i$ do not need to be CPTP, but instead they could rather be of this more general form.
This might possibly allow us to achieve a smaller extent.
Note that in order for~\Cref{prop_qpd_overhead} to still hold, we cannot weight the outcome by a factor $b_j$ larger than $1$ in absolute value, as that would increase the variance of the estimators.
For this reason, we have to restrict ourselves coefficients fulfilling $\abs{b_j}\leq 1$.

In summary, when characterizing the capabilities of a restricted quantum computer to perform quasiprobability simulation with classical side information, we should not consider the set of free CPTP maps that the computer can perform.
Instead, it is more insightful to consider the set of free quantum instruments that the computer can realize.
Given a quantum computer that can realize some set of quantum instruments $Q$, the associated decomposition set (which captures the capability of using side information) is given as follows.
\begin{definition}
  Let $Q\subset\qi(I\rightarrow
  O)$ be a set of discrete quantum instruments.
  We define the set
  \begin{align}
    \cL[Q] \coloneqq \Big\{ \sum_i a_i \cF_i :  (\cF_i)_i\in Q, a_i\in [-1,1] \Big\} \, .
  \end{align}
\end{definition}
By our above discussion, the optimal simulation overhead (with classical side information) for an operation $\cE$ is thus characterized by $\gamma_{\cL[Q]}(\cE)$.
For the rest of this section, we provide some mathematical properties of $\cL[Q]$.
\begin{definition}\label{def_q_prop}
  Let $Q\subset\qi(I\rightarrow
  O)$ be a set of quantum instruments.
  We say that
  \begin{enumerate}[(a)]
    \item $Q$ is \emph{coarsegrainable} if $\forall (\cE_i)_i\in Q$, any coarse graining of $(\cE_i)_i$ is also in $Q$.
    \item $Q$ is \emph{trivially finegrainable} if $\forall (\cE_i)_i\in Q$ we have $((1-p_i)\cE_i)_i\cup (p_i\cE_i)_i \in Q$ for any choice of probabilities $p_i\in [0,1]$.
    \item $Q$ is \emph{closed under mixture} if for any two instruments $(\cE_i)_i$ and $(\cF_j)_j$ in $Q$, their convex mixture $(p\cE_i)_i\cup ((1-p)\cF_j)_j$ is also in $Q$, for any $p\in [0,1]$.
  \end{enumerate}
\end{definition}
It is quite natural for the set of free quantum instruments to fulfill these three properties, and they hold in most physically relevant settings.
$Q$ being coarsegrainable means that one can discard classical information and treat different classical outcomes as the same.
Trivial finegranability essentially means that one can sample a random bit with the probability depending on the observed measurement outcome.
Similarly, when $Q$ is closed under mixture, then randomly choosing between two free quantum instruments is itself also a free quantum instrument in $Q$.

At first glance, using the decomposition set $\cL[Q]$ seems to add a lot of complexity, as now every element of the QPD could consist of a possibly infinite sum of quantum instrument elements.
However, this turns out not to be the case.
Indeed, it is generally sufficient to consider two-element instruments and coefficients $a_i=\pm 1$ when the set $Q$ fulfills the aforementioned natural assumptions.
\begin{lemma}\label{lem_decomp_set_characterization}
  If $Q\subset\qi(A)$ is coarsegrainable and trivially finegranable, then 
  \begin{equation}
    \cL[Q] = \{ \cE^+ - \cE^- : (\cE^+,\cE^-) \in Q\} \, .
  \end{equation}
\end{lemma}
\begin{proof}
  Consider an element $\cE = \sum_i a_i \cE_i \in \cL[Q]$ where $(\cE_i)_i\in Q$.
  By the trivially finegrainable property of $Q$, we can also express $\cE$ as a weighted sum of another instrument in $Q$, but only with coefficients $+1,-1$ and $0$:
  \begin{equation}
    \cE = \sum_i \sgn(a_i) \left(\abs{a_i} \cE_i\right) + \sum_i 0\cdot \left((1-\abs{a_i}) \cE_i\right)
  \end{equation}
  We now coarsegrain this instrument into a three-element instrument $(\cE^+,\cE^-,\cE^0)$
  \begin{align}
    \Tilde{\cE}^+ \coloneqq \sum_{i : a_i \geq 0} \abs{a_i}\cE_i\, , \qquad
    \Tilde{\cE}^- \coloneqq \sum_{i : a_i < 0} \abs{a_i}\cE_i \, , \qquad
    \cE^0 \coloneqq \sum_{i} (1-\abs{a_i})\cE_i \, ,
  \end{align}
  which is still in $Q$ by the coarsegrainability property.
  We thus have
  \begin{equation}
    \cE = \Tilde{\cE}^+ - \Tilde{\cE}^- + 0\cdot \cE^0 \, .
  \end{equation}
  Finally, the finegrainability and coarsegrainability properties imply $(\Tilde{\cE}^+,\Tilde{\cE}^-,\frac{1}{2}\cE^0,\frac{1}{2}\cE^0)$ and thus $(\Tilde{\cE}^++\frac{1}{2}\cE^0,\Tilde{\cE}^-+\frac{1}{2}\cE^0)$ are quantum instruments in $Q$.
  The desired statement thus follows from
  \begin{equation}
    \cE = (\Tilde{\cE}^++\frac{1}{2}\cE^0) - (\Tilde{\cE}^-+\frac{1}{2}\cE^0) \, .
  \end{equation}
\end{proof}

We also note that any set $Q$ being closed under mixture provides useful structure for the induced decomposition set.
\begin{lemma}\label{lem_ds_convex}
    Let $Q\subset\qi(I\rightarrow
  O)$ be a set of quantum instruments that is closed under mixture. Then $\cL[Q]$ is convex.
\end{lemma}
\begin{proof}
    Let $\cE$ and $\cF$ be two elements in $\cL[Q]$.
    By definition, we can write them as $\cE=\sum_i b_i\cE_i$ and $\cF=\sum_j c_j\cF_j$ for two quantum instruments $(\cE_i)_i$, $(\cF_j)_j$ in $Q$.
    Any convex mixture of the two elements is also in $\cL[Q]$, because we can write
    \begin{equation}
        (1-p)\cE + p\cF = \sum_i b_i (1-p)\cE_i + \sum_j c_j p\cF_j \, ,
    \end{equation}
    where $p\in [0,1]$ and $(p\cE_i)_i\cup ((1-p)\cF_j)_j$ is in $Q$ due to it being closed under mixture.
\end{proof}
In fact, if a set of quantum instruments $Q$ is closed under mixture, then the induced set of quantum channels is also convex.
\begin{lemma}\label{lem_s_star_convex}
    Let $Q\subset\qi(I\rightarrow O)$ be a set of quantum instruments and $S^{\star}\coloneqq \cL[Q]$.
    We denote by $S\coloneqq \{\sum_i\cE_i | (\cE_i)_i\in Q\}$ the set of quantum channels induced by $Q$.
    If $S^{\star}$ is convex, then so is $S$.
\end{lemma}
\begin{proof}
  Consider two channels $\cE,\cF\in S$.
  We know that any convex mixture $(1-p)\cE +p\cF$ lies in $S^{\star}$, and we now need to show that it also lies in $S$.
  By the definition of $S^{\star}$, we can write
  \begin{equation}
    (1-p)\cE + p\cF = \sum_j b_j \cG_j
  \end{equation}
  for some $b_j\in [-1,1]$ and $(\cG_j)_j\in Q$.
  Taking the Choi isomorphism on both sides, as well as the partial trace $\tr_O$ thereof, we obtain
  \begin{equation}
    \frac{1}{\dim(I)}\id = \sum_j b_j \tr_O[\cJ(\cG_i)] \, .
  \end{equation}
  At the same time, we know that $\sum_j\cG_j\in\CPTP(I\rightarrow O)$, so
  \begin{equation}
    \frac{1}{\dim(I)}\id = \sum_j b_j \tr_O[\cJ(\cG_j)] \leq \sum_j \tr_O[\cJ(\cG_j)] = \frac{1}{\dim(I)}\id \, .
  \end{equation}
  This implies that for all $j$ we have $b_j=1$, unless $\cG_j$ is itself zero.
  Therefore, $(1-p)\cE + p\cF\in S$.
\end{proof}

\section{Quasiprobabilistic circuit cutting}\label{sec:circuit_cutting}
In this section, we elaborate how the quasiprobability technique can be used to simulate a large quantum computer with a small one.
Suppose one wants to run a quantum circuit on a given quantum device, but its number of available qubits is too small to fit the circuit.
The premise of \emph{circuit cutting} is to subdivide the circuit into smaller partitions which are each small enough to fit on the available quantum computer.
The goal is to reconstruct the outcome of the original circuit by executing the smaller partitions on a small quantum computer and performing appropriate classical post-processing.

It is natural to realize circuit cutting through quasiprobabilty simulation.
Consider for instance the setup depicted in \Cref{fig:cutting_example}.
The qubits in a circuit are grouped into two systems $A$ and $B$.
Every gate that acts nonlocally across $A$ and $B$ is considered a non-free operation, and we will simulate them using operations that only act locally on each partition.

\begin{figure}[!htb]
    \centering
    \input{figure_cutting_example}
    \caption{The nonlocal circuit on the left can be simulated with local circuits on the right using quasiprobability simulation. Each non-local gate $U_1,U_2$ and $U_3$ is simulated using local operations.}
    \label{fig:cutting_example}
\end{figure}
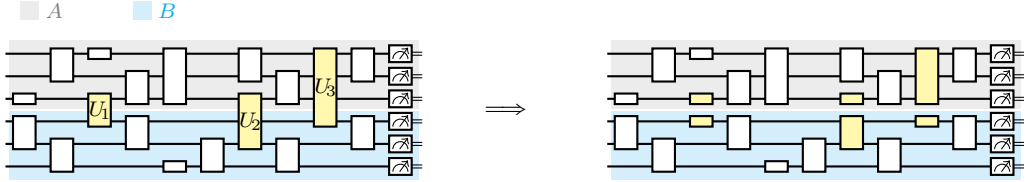

It is common to distinguish between two settings where classical communication between the subpartitions is either allowed or not allowed.
Mathematically, the distinction between these two situations is captured by choosing the decomposition set $S$ to be $\LO$ (``local operations'') or $\LOCC$ (``local operations with classical communication''), respectively, $\LO^{\star}$ and $\LOCC^{\star}$ when classical side information is also allowed.
A precise mathematical definition of these four sets will be given further below.

Classical communication could potentially lead to a lower simulation overhead for certain bipartite channels $\cE$. In~\cref{app:remaining_inequalities}, we show that there exist bipartite channels $\cE$ such that $\gamma_{\LOCC^{\star}}(\cE) < \gamma_{\LO^{\star}}(\cE)$. 
It is known that $\gamma_{\LOCC^{\star}}(\cU) = \gamma_{\LO^{\star}}(\cU)$ for Clifford gates~\cite{piv23}, all two-qubit unitary channels $\cU$ and KAK-like generalizations thereof~\cite{SPS24}.
It is unknown if this equality is true for arbitrary unitaries.

The use of classical communication comes at an additional price in terms of practical implementation costs. Without classical communication, the individual subcircuits could be run sequentially and in principle on the same quantum device.
However, once classical communication is involved, multiple small quantum processors will generally be required to execute the subcircuits simultaneously.

The set $\LOCC$ is notoriously hard to work with due to its intricate mathematical structure.
A common strategy to ease this difficulty is to replace $\LOCC$ with a relaxation of separable channels $\SEP$ or positive partial transpose channels $\PPT$, which are strict supersets of $\LOCC$ with simpler characterization.
A precise mathematical definition will follow later, but we already point out that
\begin{equation}\label{eq:chan_hierarchy}
    \LO \subset \LOCC \subset \SEP \subset \PPT
\end{equation}
and
\begin{equation}\label{eq:chanstar_hierarchy}
    \LO^{\star} \subset \LOCC^{\star} \subset \SEP^{\star} \subset \PPT^{\star} \, ,
\end{equation}
which in turn implies the relations between the associated quasiprobability extents depicted in~\cref{fig_simple}.\footnote{When considering the extent of some bipartite channel $\cE\in\cptp(AB\rightarrow A'B')$, we will typically drop the systems $AB$ in the subscript (e.g. $\gamma_{\LOCC}(\cE)$ instead of $\gamma_{\LOCC(AB\rightarrow A'B')}(\cE)$) to simplify the notation. The implied systems will be clear from context.}
We highlight that compared to $\LO$ and $\LOCC$, $\SEP$ and $\PPT$ have no operational meaning for the purpose of circuit cutting and instead only serve as a mathematical tool to find lower bounds.
\begin{figure}[!htb]
    \centering
    \input{simple_gamma_fig}
\caption{Trivial relations of extents for different settings. The inequalities follow from the inclusion relations of the different decomposition sets (see \Cref{eq:chan_hierarchy,eq:chanstar_hierarchy}). A main result of this work is a more refined version of this diagram, which is depicted in \Cref{fig_results}.}
\label{fig_simple}
\end{figure}

We now introduce the quantum instruments of these four settings in detail:
\begin{align}
    \qi_{\LO(AB\rightarrow A'B')} &\coloneqq \big \{ (\cE_i\otimes\cF_j)_{i,j} : (\cE_i)_i\in\qi(A\rightarrow A'), (\cF_j)_j\in\qi(B\rightarrow B') \big \} \\
    \qi_{\SEP(AB\rightarrow A'B')} &\coloneqq  { \big \{ (\cE_i)_i\in \qi(A B\rightarrow A'B') : \forall i: \frac{\cJ(\cE_i)}{\tr[\cJ(\cE_i)]} \in \SEP(AA',BB')  \big\} }\\
    \qi_{\PPT(AB\rightarrow A'B')} &\coloneqq \big \{ (\cE_i)_i\in \qi(A B \rightarrow A'B') : \forall i: \cJ(\cE_i)^{T_B}\geq 0 \big \} \, .
\end{align}
The precise definition of the $\LOCC$ instruments $\qi_{\LOCC}$ is technical and not necessary for our purposes here. We refer the reader to~\cite{CLMOW14} for a precise mathematical treatment of $\qi_{\LOCC}$.
All these instruments induce a set of achievable channels, i.e., 
\footnotesize
\begin{align}
     \LO(AB\rightarrow A'B') & \coloneqq \Big \{ \sum_i\cE_i : (\cE_i)_i\in\qi_{\LO(AB\rightarrow A'B')} \Big \}  = \big \{ \cE\otimes\cF : \cE\in\CPTP(A \rightarrow A'), \cF\in\CPTP(B \rightarrow B') \big \} \\
     \LOCC(AB\rightarrow A'B') & \coloneqq \Big \{ \sum_i\cE_i : (\cE_i)_i\in\qi_{\LOCC(AB\rightarrow A'B')} \Big \} \\
     \SEP(AB\rightarrow A'B') & \coloneqq \Big \{ \sum_i\cE_i \!:\! (\cE_i)_i\in\qi_{\SEP(AB\rightarrow A'B')} \Big \} \!=\! \big \{ \cE\in\CPTP(AB\rightarrow A'B') \!:\! \cJ(\cE)\in\SEP(AA',BB') \} \big \} \\
     \PPT(AB\rightarrow A'B') & \coloneqq \Big \{ \sum_i\cE_i \!:\! (\cE_i)_i\!\in\!\qi_{\PPT(AB\rightarrow A'B')} \Big \}  \!=\! \big \{ \cE\!\in\!\CPTP(AB\rightarrow A'B') \!:\! \cJ(\cE)\in\PPT(AA',BB') \} \big \} \, .
\end{align}
\normalsize
Following the procedure from~\Cref{sec_intermediate_meas}, the corresponding decomposition sets capturing the notion of classical side information are given by
\begin{align}
     \LO^{\star}(AB\rightarrow A'B') & \coloneqq \cL[\qi_{\LO(AB\rightarrow A'B')}] \\
     \LOCC^{\star}(AB\rightarrow A'B') & \coloneqq \cL[\qi_{\LOCC(AB\rightarrow A'B')}] = \{ \cE^+-\cE^- : (\cE^+,\cE^-)\in\qi_{\LOCC(AB\rightarrow A'B')}\} \\
     \SEP^{\star}(AB\rightarrow A'B') & \coloneqq \cL[\qi_{\SEP(AB\rightarrow A'B')}] = \{ \cE^+-\cE^- : (\cE^+,\cE^-)\in\qi_{\SEP(AB\rightarrow A'B')}\}  \\
     \PPT^{\star}(AB\rightarrow A'B') & \coloneqq \cL[\qi_{\PPT(AB\rightarrow A'B')}] = \{ \cE^+-\cE^- : (\cE^+,\cE^-)\in\qi_{\PPT(AB\rightarrow A'B')}\} \, ,
\end{align}
where we used~\Cref{lem_decomp_set_characterization} for the characterization of $\LOCC^{\star}, \SEP^{\star}$ and $\PPT^{\star}$. 
Indeed, it follows immediately from the definition of $\qi_{\SEP}$ and $\qi_{\PPT}$ that  they are coarsegrainable and trivially finegrainable.
For $\qi_{\LOCC}$, the coarsegrainability follows directly from its definition and trivial finegrainability essentially can be seen by concatenating any $\LOCC$ protocol with one of the two parties choosing a random bit with probability based on all previously obtained measurement outcomes.
Furthermore, $\LOCC^{\star}, \SEP^{\star} $ and $\PPT^{\star}$ are convex by~\Cref{lem_ds_convex} because the corresponding sets of quantum instruments are closed under mixture.\footnote{For $\LOCC$ one can use one round of classical communication to randomly choose between one of two protocols to perform.}
These properties do unfortunately not hold for $\qi_{\LO}$, making its characterization slightly more challenging -- this will not be a problem for our purposes in this paper.

The four proposed settings constitute a hierarchy in the sense that
\begin{equation}
  \qi_{\LO} \subset \qi_{\LOCC} \subset \qi_{\SEP} \subset \qi_{\PPT} \, ,
\end{equation}
which therefore implies \Cref{eq:chan_hierarchy,eq:chanstar_hierarchy} and the relations in \Cref{fig_simple}.
One of the main technical results of this work is to refine these relations with those seen in ~\Cref{fig_results}.

We highlight that the two quantities with practical relevance for circuit cutting are $\gamma_{\LO^{\star}}$ and $\gamma_{\LOCC^{\star}}$.
The quantities $\gamma_{\LO}$ and $\gamma_{\LOCC}$ are not as natural, as there is little reason not to use the power classical side information, especially in the $\LOCC$ setting where intermediate measurements are required in the $\LOCC$ protocol itself.
We still study these two quantities to understand the improvement that classical side information can enable.
The quantities $\gamma_{\PPT},\gamma_{\SEP},\gamma_{\PPT^{\star}}$ and $\gamma_{\SEP^{\star}}$ have no physical meaning and will be used purely as a mathematical tool to provide lower bounds for the other quantities.

\section{Utility of classical side information} \label{sec_QI}
In this section, we study whether the usage of classical side information allows for a smaller sampling overhead.
Before we consider general channels, let us consider the simpler case of states.
Here, classical side information does not provide any utility, and in fact, we get the same extent in the $\LOCC,\LO$ and $\SEP$ settings.
\begin{lemma} \label{lem_extent_states}
    For any bipartite density matrix $\rho \in \states(AB)$ we have
    \begin{equation}
        \gamma_{\LO}(\rho) 
        = \gamma_{\LO^{\star}}(\rho) 
        = \gamma_{\LOCC}(\rho) 
        = \gamma_{\LOCC^{\star}}(\rho) 
        = \gamma_{\SEP}(\rho) 
        = \gamma_{\SEP^{\star}}(\rho) \, .
    \end{equation}
\end{lemma}
Recall that in this context, states should be understood as quantum channels with trivial input space, i.e., $\states(AB)\equiv \cptp(\mathbb{C}\rightarrow AB)$.
Note that the separable quantum channels with trivial input space are precisely equivalent to the set of separable states, i.e., $\SEP(\mathbb{C}\otimes\mathbb{C}\rightarrow AB) \equiv \SEP(A,B)$.
\begin{proof}
    Considering the chain of inequalities in \Cref{fig_simple}, it suffices to show that $\gamma_{\LO}(\rho) \leq \gamma_{\SEP^{\star}}(\rho)$.
    For this purpose, it suffices to show that any QPD of $\rho$ with respect to $\SEP^{\star}$ induces a QPD of $\rho$ with respect to $\LO$ with the identical $\ell_1$-norm of its coefficients.

    For this purpose, note that $\LO(\mathbb{C}\otimes\mathbb{C}\rightarrow AB)$ precisely corresponds to the set of states of the form $\tau\otimes\kappa$ where $\tau\in\states(A),\kappa\in\states(B)$.
    Similarly, for $\gamma_{\SEP^\star}$, we have to consider the set $\SEP^{\star}(\mathbb{C}\otimes\mathbb{C}\rightarrow AB)$. By \Cref{lem_decomp_set_characterization}, it is given by
    \begin{equation}
        \SEP^{\star}(\mathbb{C}\otimes\mathbb{C}\rightarrow AB) \equiv \{ (1-p)\sigma^+ - p\sigma_- : p\in [0,1], \sigma^{\pm}\in \SEP(A,B)\} \, .
    \end{equation}

    Let $\rho=\sum_i a_i\sigma_i$ be a QPD where $\sigma_i\in\SEP^{\star}$, i.e., we can write $\sigma_i = (1-p_i)\sigma_i^+ - p_i\sigma_i^-$ for $p_i\in[0,1],\sigma_i^{\pm}\in\SEP(A,B)$.
    Thanks to separability, we can write $\sigma_i^{\pm} = \sum_j q_j^{i,\pm} \tau_j^{i,\pm}\otimes \kappa_j^{i,\pm}$ where $(q_j^{i,\pm})_j$ is a probability distribution and $\tau_j^{i,\pm}\in\states(A),\kappa_j^{i,\pm}\in\states(B)$.
    We can then rewrite
    \begin{equation}
        \rho = \sum_{i,j} a_i(1-p_i)q_j^{i,+} \tau_j^{i,+}\otimes \kappa_j^{i,+}  -  \sum_{i,j} a_ip_iq_j^{i,-} \tau_j^{i,-}\otimes \kappa_j^{i,-}
    \end{equation}
    which is a QPD with respect to $\LO(\mathbb{C}\otimes\mathbb{C}\rightarrow AB)$ with $\ell_1$-norm of its coefficients given by
    \begin{equation}
        \sum_{i,j} \abs{a_i(1-p_i)q_j^{i,+}} + \abs{a_ip_iq_j^{i,-}} \leq \sum_i \abs{a_i} \, .
    \end{equation}
\end{proof}

The next result states that in the $\LO$ setting, entangling maps cannot be simulated at all without classical side information.
Indeed, the only simulable maps in this setting are non-signalling maps, i.e., maps that do not allow for the transmission of any kind of information.
In contrast, by using classical side information all Hermitian preserving superoperators can be achieved.\footnote{ Note that $\cptp(AB\rightarrow A'B')\subset\HP(AB\rightarrow A'B')$.}
\begin{definition}
    A quantum operation $\cE\in\HP(AB\rightarrow A'B')$ is said to be non-signalling if
    \begin{equation}
        \forall\rho,\sigma \in\states(AB) \quad \textnormal{with} \quad \tr_B[\rho]=\tr_B[\sigma] \quad \textnormal{we have} \quad \tr_B[\cE(\rho)]=\tr_B[\cE(\sigma)] \, .
    \end{equation}
    We denote the set of such non-signalling operation by $\mathrm{NSIG}(AB\rightarrow A'B')$.
\end{definition}

\begin{proposition} \label{prop_LO_vs_LO_star}
    The set of superoperators that exhibit a QPD under $\LO$ and $\LO^{\star}$ are given by $\mathrm{span}_{\mathbb{R}}\big(\LO(AB\rightarrow A'B')\big) = \mathrm{NSIG}(AB\rightarrow A'B')$ and $\mathrm{span}_{\mathbb{R}}\big(\LO^{\star}(AB\rightarrow A'B')\big)=\HP(AB\rightarrow A'B')$, respectively.
\end{proposition}
\begin{proof}
    We start by characterizing the span of $\LO(AB)$.
    It can be shown that
    \begin{equation}
        \mathrm{NSIG}(AB\rightarrow A'B')=\left \lbrace \sum_{i,j}\alpha_{ij} \cF_i\otimes \cG_j : \alpha_{ij}\in\mathbb{R},\cF_i\in\HPTP(A\rightarrow A'),\cG_j\in\HPTP(B\rightarrow B') \right \rbrace \, .
    \end{equation}
    For a proof, we refer to~\cite[Theorem 13]{gutoski09} or~\cite[Theorem 5.1]{Cavalcanti_2022}.
    Clearly, \smash{$\mathrm{span}_{\mathbb{R}}\{\LO(AB\rightarrow A'B')\}$} is contained in this set.
    To show the converse, it suffices to argue that any $\cF\otimes \cG$ for $\cF\in\HPTP(A\rightarrow A'),\cG\in\HPTP(B\rightarrow B')$ lies in $\mathrm{span}_{\mathbb{R}}\{\LO(AB\rightarrow A'B')\}$.
    This is a direct consequence of a slightly simpler result: Any $\HPTP$ map $\cE\in\HPTP(I\rightarrow O)$ can be decomposed into a QPD of two $\CPTP$ maps.
    Applying this to $\cF$ and $\cG$ provides the desired statement.
    One such valid QPD of $\cE$ is given by 
    \begin{equation}
        \cE = d_{IO}\lambda \cD -(d_{IO}\lambda -1)\left( \frac{d_{IO}\lambda}{d_{IO}\lambda-1}\cD - \frac{1}{d_{IO}\lambda-1}\cE \right) \, ,
    \end{equation}
    where $\cD(\rho)$ is the fully depolarizing channel from $I$ to $O$ and $d_{IO}\coloneqq\dim(I)\dim(O)$ and $\lambda\in\mathbb{R}$ is at least the maximal eigenvalue of $\cJ(\cE)$.
    Indeed, $( \frac{d_{IO}\lambda}{d_{IO}\lambda-1}\cD - \frac{1}{d_{IO}\lambda-1}\cE)$ can easily be checked to be CPTP as
    \begin{equation}
      \tr_O\left[ \frac{d_{IO}\lambda}{d_{IO}\lambda-1}\cJ(\cD) - \frac{1}{d_{IO}\lambda-1}\cJ(\cE) \right] = \left(\frac{d_{IO}\lambda}{d_{IO}\lambda-1}-\frac{1}{d_{IO}\lambda-1}\right)\frac{1}{d_{I}}\id_{I} = \frac{1}{d_{I}}\id_{I}
    \end{equation}
    and
    \begin{align}
      & \left( \frac{d_{IO}\lambda}{d_{IO}\lambda-1}\cJ(\cD) - \frac{1}{d_{IO}\lambda-1}\cJ(\cE) \right) \geq 0\, ,
    \end{align}
    where the final step is equivalent to $d_{IO}\lambda \frac{1}{d_{IO}}\id_{IO} \geq \cJ(\cE)$.

    Next, we consider the span of $\LO^{\star}(AB)$.
    Clearly, by definition we have
    \begin{equation}
        \mathrm{span}_{\mathbb{R}}(\LO^{\star}(AB))\subset \HP(AB\rightarrow A'B') \, ,
    \end{equation}
    as any real linear combination of Hermitian preserving maps is itself Hermitian preserving, so it suffices to show the converse.
    For this purpose, notice that
    \begin{equation}
        \HP(AB\rightarrow A'B') = \left \lbrace \sum_{i,j}\cF_i\otimes\cG_j : \cF_i\in\HP(A\rightarrow A'), \cG_j\in\HP(B\rightarrow B')\right \rbrace\ \, .
    \end{equation}
    In the following, we will argue that any $\cE\in\HP(I\rightarrow O)$ can be written as a QPD of two elements in $\mathcal{L}[\qi(I\rightarrow O)]$.
    Applying that to each $\cF_i$ and $\cG_j$ implies the desired statement.
    We can split the Choi representation of $\cE$ into a difference of two positive operators $\Lambda^{\pm}$
    \begin{equation}
        \cJ(\cE) = \Lambda^+ - \Lambda^- = \lambda^+\frac{\Lambda^+}{\lambda^+} - \lambda^-\frac{\Lambda^-}{\lambda^-} \, ,
    \end{equation}
    where we denote by $\lambda^{\pm}$ the maximal eigenvalues of $\dim(I)\tr_O[\Lambda^{\pm}]$.
    We can observe that $\Lambda^{\pm}/\lambda^{\pm}$ are Choi operators representing completely positive trace-non-increasing maps, as they are positive and $\tr_O[\Lambda^{\pm}/\lambda^{\pm}] \leq \dim(I)^{-1}\id_I$.
    The desired statement follows immediately from $\CPTN(I\rightarrow O)\subset\cL[\qi(I\rightarrow O)]$.
    Note that we implicitly assumed that $\lambda^{\pm}>0$, but the proof can be straightforwardly adapted for the case where $\cJ(\cE)$ is positive or negative.
\end{proof}
Interestingly, the situation is different for the $\LOCC$, $\SEP$ and $\PPT$ settings.
Indeed, here it is possible to simulate any trace-preserving map without using classical side information.
\begin{proposition}
    The set of superoperators that exhibit a QPD under $\LOCC,\LOCC^{\star},\SEP,\SEP^{\star},\PPT$ and $\PPT^{\star}$ are given by\footnote{Note that here we omit the input systems $A,B$ and output systems $A',B'$ to improve readability. For instance, $\LOCC$ stands for $\LOCC(AB\rightarrow A'B')$.}
    \begin{align}
        \mathrm{span}_{\mathbb{R}}(\LOCC)
        =\mathrm{span}_{\mathbb{R}}(\SEP)
        = \mathrm{span}_{\mathbb{R}}(\PPT) 
        = \mathbb{R}\HPTP \label{eq_eq_1_prop}
    \end{align}
    and
    \begin{align}
        \mathrm{span}_{\mathbb{R}}(\LOCC^{\star})
        = \mathrm{span}_{\mathbb{R}}(\SEP^{\star})
        = \mathrm{span}_{\mathbb{R}}(\PPT^{\star})
        = \HP \,, \label{eq_eq_2_prop}
    \end{align}
    where $\mathbb{R}\HPTP\coloneqq \{r\cdot\cE : r\in\mathbb{R}, \cE\in\HPTP\}$.
\end{proposition}
\begin{proof}
    We start by proving~\cref{eq_eq_2_prop}.
    Clearly, one has
    \begin{equation}
        \mathrm{span}_{\mathbb{R}}(\LOCC^{\star}) 
        \subset \mathrm{span}_{\mathbb{R}}(\SEP^{\star})
        \subset \mathrm{span}_{\mathbb{R}}(\PPT^{\star})
        \subset \HP 
    \end{equation}
    as any (real) linear combination of Hermitian preserving maps it itself Hermitian perserving.
    The converse is a direct consequence from~\Cref{prop_LO_vs_LO_star}, i.e. 
    \begin{equation}
        \HP = \mathrm{span}_{\mathbb{R}}(\LO^{\star}) \subset \mathrm{span}_{\mathbb{R}}(\LOCC^{\star}) \, .
    \end{equation}

    Next, we prove~\cref{eq_eq_1_prop}.
    Again, one has
    \begin{equation}
        \mathrm{span}_{\mathbb{R}}(\LOCC) \subset \mathrm{span}_{\mathbb{R}}(\SEP) \subset \mathrm{span}_{\mathbb{R}}(\PPT) \subset \mathbb{R}\HPTP \, ,
    \end{equation}
    as any (real) linear combination of trace-preserving maps is itself propotional to a trace-preserving map.
    It thus suffices to show $\HPTP\subset \mathrm{span}_{\mathbb{R}}(\LOCC)$, which then implies $\mathbb{R}\HPTP\subset \mathrm{span}_{\mathbb{R}}(\LOCC)$.
    We also know from the proof of \Cref{prop_LO_vs_LO_star} that any $\HPTP$ map can be decomposed as a QPD of CPTP maps.
    Hence, it suffices to show that any $\cE\in\CPTP(AB\rightarrow A'B')$ lies in $\mathrm{span}_{\mathbb{R}}(\LOCC)$.
    We assume without loss of generality that $A=A'$ and $B=B'$ and that both are qubit systems (otherwise, we can always embed these systems into a large enough qubit system).
    Consider the Stinespring dilation of $\cE$
    \begin{equation}
        \cE(\rho) = \tr_E [U_{AEB}(\rho\otimes\proj{0}_E)U_{AEB}^{\dagger}] \, .
    \end{equation}
    It suffices to show that the unitary channel $\cU\in\CPTP(A E B)$ induced by $U_{AEB}$ can be decomposed into channels in $\LOCC(AE B)$ (here, the bipartition is chosen to be across $AE$ and $B$).
    The unitary $U_{AEB}$ can be decomposed into single-qubit gates and CNOT gates~\cite[Section~4.5.3]{nielsenChuang_book}. 
    All CNOT gates that act non-locally between e.g. $AE$ and $B$ can be realized using CNOT gate teleportation~\cite{gate_tel_99}, i.e., using LOCC and consuming a pre-shared entangled Bell pair
    \begin{equation}
        \cU(\rho_{AEB}) = \cF(\rho_{AEB}\otimes\proj{\Psi}_{C_A C_B}) \, ,
    \end{equation}
    where $\cF\in\LOCC(A E  C_AC_B   B)$ (here, the bipartition is across $AEC_A$ and $C_BB$) and $\proj{\Psi}_{C_A C_B}$ is a collection of Bell pairs between the ancilla systems $C_A$ and $C_B$.
    Since there exists a QPD of the Bell pairs
    \begin{equation}
        \proj{\Psi}_{C_AC_B} = a^+\sigma^+ - a^-\sigma^- \, ,
    \end{equation}
    into separable states $\sigma^{\pm}\in\SEP(C_AC_B)$ (see e.g.~\cite{vidal99,piv23}), we can thus decompose $\cU$ into two $\LOCC(AEB)$ maps
    \begin{equation}
        \cU(\rho_{AEB}) = a^+\cF(\rho_{AEB}\otimes\sigma^{+}_{C_AC_B}) - a^-\cF(\rho_{AEB}\otimes\sigma^{-}_{C_AC_B}) \, .
    \end{equation}
\end{proof}

So any map that is trace-preserving (or proportional to a trace-preserving map) can be simulated without classical side information.
Note that this does not mean that the extent with and without side information are equal -- a priori the quasiprobability simulation with classical side information could be strictly more efficient.
We next show that this is not the case in the $\SEP$ and $\PPT$ settings.
\begin{proposition} \label{prop_PPT_PTT_star}
  Let $\cE$ be a $\HPTP$ map. Then, $\gamma_{\PPT}(\cE) = \gamma_{\PPT^\star}(\cE)$ and $\gamma_{\SEP}(\cE) = \gamma_{\SEP^\star}(\cE)$.
\end{proposition}
\Cref{prop_PPT_PTT_star} is a consequence of the following Lemma. 
\begin{lemma}\label{lemma_s_s_star}
    Let $Q$ be a set of quantum instruments that fulfills the three properties in~\cref{def_q_prop} and $\forall (\cE_i)_i\in Q$ and for each element $\cE_i$ that is proportional to a trace-preserving map, the rescaled map $\frac{1}{\tr[\cE_i]}\cE_i$ is itself a (single-element) quantum instrument in $Q$.
    Let $S\coloneqq \{ \sum_i\cE_i : (\cE_i)_i\in Q\}$ and $S^{\star}\coloneqq\cL[Q]$.
    Then, $\gamma_{S}(\cE)=\gamma_{S^{\star}}(\cE)$ for all HPTP maps $\cE$.
\end{lemma}

As discussed in the previous section, the quantum instruments of the sets $\PPT$ and $\SEP$ fulfill the three properties in \Cref{def_q_prop}.
Furthermore, an element of any SEP/PPT instrument has a separable/PPT Choi state and can hence be rescaled to a SEP/PPT channel.
Therefore, both $\PPT$ and $\SEP$ fulfill the requirements of~\cref{lemma_s_s_star}.

Note that $\LOCC$ also satisfies the three properties in \Cref{def_q_prop}, but does not fulfill the rescaling property.
In fact, while $\LOCC\subsetneq\SEP$, is known that any channel in $\SEP$ (or in fact any separable instrument) can be \emph{stochastically} realized with LOCC.
More precisely, for any $\cE\in\SEP$ there exists some probability $p>0$ such that $(p\cE,(1-p)\cD)\in\qi_{\LOCC}$ where $\cD$ is the fully depolarizing channel~\cite{D_r_2000,CLMOW14}.

We conclude that the technique from \Cref{lemma_s_s_star} is not applicable to the $\LOCC$ setting.
However, this does not rule out the possibility that other $\LOCC$ constructions could exist which achieve the same extent as the $\LOCC^\star$ construction.

\begin{proof}[Proof of~\cref{lemma_s_s_star}]
The relation $S \subseteq S^\star$ implies $\gamma_{S}(\cE) \geq \gamma_{S^\star}(\cE)$. Hence it suffices to show that given any QPD of $\cE$ with respect to $S^\star$, we can construct a QPD with respect to $S$ with at most the same $\ell_1$-norm of its coefficients.
Since Q is closed under mixture, we know by~\cref{lem_ds_convex} and~\cref{lem_s_star_convex} that both $S$ and $S^\star$ are convex. 
Therefore, by \Cref{lemma_convex} we can assume without loss of generality, that any QPD with respect to $S^{\star}$ has two elements.

Let
\begin{align}
    \cE = a_+ \cF_+ - a_- \cF_- \, .
\end{align}
be a QPD w.r.t. $S^{\star}$.
By~\Cref{lem_decomp_set_characterization}, $\cF_\pm \in S^\star$ can be written as
\begin{equation}
    \cF_\pm = \cG_\pm^+ -\cG_\pm^- \, ,
\end{equation}
for some instruments $(\cG_+^+,\cG_+^-)$ and $(\cG_-^+,\cG_-^-)$ that lie in $Q$.
We can regroup this expression as
\begin{align}
    \cE = a_+ \cF_+ - a_- \cF_- 
    = a_+ (\cG_+^+ - \cG_+^-) - a_- (\cG_-^+ - \cG_-^-)
    = (a_+ \cG_+^+ + a_- \cG_-^-)  - (a_+ \cG_+^-  + a_- \cG_-^+) \, . \label{eq:decomp_ppt}
\end{align}
Now we show that these two grouped terms are proportional to a valid map in $S$.

First, we show that they are proportional to a $\CPTP$ map. Being the sum of quantum instruments, complete positivity is clear, so it remains to show that the map is trace preserving.
For this, we consider the reduced Choi state. By naming the target system of the channels $C$, let us introduce
\begin{align}
    \cM = \tr_C[ \cJ ( \cG_+^+ ) ] \qquad \textnormal{and} \qquad
    \cN = \tr_C[\cJ( \cG_-^-)] \,.
\end{align}
Since the instruments fulfill $\cG_{\pm}^+ + \cG_{\pm}^- \in \CPTP$, we find for the other two instruments
\begin{align}
     \tr_C[\cJ(\cG_+^-)]  = \frac{\id}{d} - \cM \qquad \textnormal{and} \qquad
     \tr_C[\cJ(\cG_-^+)] = \frac{\id}{d} - \cN\, ,
\end{align}
where $d$ is the dimension of the input system and $\frac{\id}{d}$ describes the maximally mixed state.
Applying the partial trace to \cref{eq:decomp_ppt}, we get
\begin{align}
    \frac{\id}{d} &= a_+ \cM + a_- \cN - \big(a_+(\frac{\id}{d} -\cM) + a_- (\frac{\id}{d} - \cN)\big) \, ,
\end{align}
which can be rewritten as
\begin{align} \label{eq_lukas_1}
  \frac{1+a_+ +a_-}{2} \frac{\id}{d} = a_+ \cM + a_- \cN \, . 
\end{align}
Hence, $\cP_+ \coloneqq \frac{2}{1+a_+ +a_-}(a_+ \cG^+_+ + a_- \cG_{-}^{-})$ is  a CPTP map. 
Since Q is trivially finegrainable, the following instrument is an element of Q:
\begin{equation}
    \left(\frac{a_+}{a_++ a_-}\cG^+_+,\frac{a_+}{a_++ a_-}\cG^-_+,\frac{a_-}{a_++ a_-}\cG^+_-,\frac{a_-}{a_++ a_-}\cG^-_-\right) \, .
\end{equation}
Since Q is coarsegrainable, and $\cP_+$ is proportional to $\frac{a_+}{a_++ a_-}\cG^+_+ + \frac{a_-}{a_++ a_-}\cG^-_-$ we can use the rescaling assumption to conclude that $\cP_+$ is in $S$.
Analogously, we find
    \begin{align}
        \tr_C[a_+ \cJ(\cG_+^-)  + a_- \cJ( \cG_-^+)] 
        &= a_+ (\frac{\id}{d} -\cM) + a_-(\frac{\id}{d} -\cN) \\
        &= (a_+ + a_-)\frac{\id}{d} - (a_+ \cM + a_-\cN) \\
        &= (a_+ + a_-)\frac{\id}{d} - \frac{1+a_++a_-}{2} \frac{\id}{d} \\
        &= \frac{a_+ + a_- -1}{2} \frac{\id}{d} \, ,
    \end{align}
where the third equation uses~\cref{eq_lukas_1}.    
Thus, $\cP_- \coloneqq \frac{2}{a_+ +a_- - 1}(a_+ \cG_+^- + a_- \cG_-^+)$ is in $S$ as well.   
With this, we found a decomposition of $\cE$ into maps in $S$ of the form
\begin{align}
         \cE = \frac{1+a_+ +a_-}{2} \cP_+ - \frac{a_+ + a_- -1}{2} \cP_-
\end{align}
and therefore
\begin{align}
\gamma_{S}(\cE) \leq  \frac{1+a_+ +a_-}{2} + \frac{a_+ + a_- -1}{2} = a_+ + a_-  \, .
\end{align}    

\end{proof}

\section{PPT lower bound}\label{sec_ppt}
\Cref{prop_PPT_PTT_star} implies that the quantity $\gamma_{\PPT}$, which was introduced in~\cite{xin24}, does not only bound $\gamma_{\LO}$ and $\gamma_{\LOCC}$ from below, but also the physically more relevant quantities $\gamma_{\LO^{\star}}$ and $\gamma_{\LOCC^{\star}}$, respectively.
For any $\cE\in\TP(A B)$ we have
\begin{align} \label{eq_PPT_LB}
\gamma_{\LO^\star}(\cE) \geq \gamma_{\LOCC^\star}(\cE) \geq \gamma_{\PPT^\star}(\cE)  \overset{\textnormal{\Cshref{prop_PPT_PTT_star}}}{=} \gamma_{\PPT}(\cE) \, .
\end{align}
This bound is interesting for two reasons:
\begin{enumerate}[(i)]
\item It can be evaluated efficiently via semidefinite programming~\cite[Equation~B2]{xin24}. Let $\cJ(\cE)_{AA'BB'}$ denote the Choi state of $\cE\in \HPTP(AB\rightarrow A'B')$, then 
    \begin{align}  \label{eq_PPT_SDP}
      \gamma_{\PPT}(\cE)\! =\! \min \limits_{c\geq 0, X_{AA'BB'}}\big\{ 2c \!-\!1 : X^{\mathrm{T}_{BB'}} \!\geq\! 0 , X^{\mathrm{T}_{BB'}} \geq \cJ(\cE)^{\mathrm{T}_{BB'}}, X_{AB} \!=\! c\, \id_{AB}, X \!\geq\! \cJ(\cE) \big \} \, ,
    \end{align}
    which is a semidefinite program (SDP).
\item As we will see in the below, the lower bound from~\cref{eq_PPT_LB} is tight for many instances of $\cE$.
\end{enumerate}

The remainder of this section discusses the $\PPT$ lower bound from~\cref{eq_PPT_LB} and its consequences in more detail.
First, we show that for the special case of pure states, the $\PPT$ lower bound is tight (recall that for states, $\gamma_{\LOCC}=\gamma_{\LO}=\gamma_{\SEP}$).
\begin{proposition}\label{prop_ppt_pure_states}
    For any pure bipartite state $\ket{\psi}_{AB}$ one has $\gamma_{\SEP}(\proj{\psi})=\gamma_{\PPT}(\proj{\psi})$.
\end{proposition}
\begin{proof}
   Since $\SEP \subseteq \PPT$ we have $\gamma_{\SEP}(\proj{\psi})\geq \gamma_{\PPT}(\proj{\psi})$. Hence, it remains to show the converse.
    Let's introduce a quantity analogous to the robustness of entanglement, but relative to $\PPT$ instead of $\SEP$, i.e.,
    \begin{align}
        R_{\PPT}(\rho_{AB})\coloneqq \min_{\sigma_{AB} \in \PPT(AB),t\geq 0} \Big \{ t : \frac{\rho_{AB}+t \sigma_{AB}}{1+t} \in \mathrm{PPT}\Big \} \, .
    \end{align}
    This quantity trivially fulfills $R_{\PPT}(\ket{\psi}) \leq R_{\SEP}(\ket{\psi})$. In the following, we show the opposite direction, which implies the proposition by \Cref{lem_robustness}. 
    To prove the desired statement, we must show that there exists no $\sigma \in \St$, such that $t<R_{\SEP}(\ket{\psi})$ and
    \begin{equation}
        \frac{1}{1+t}\left( \ket{\psi}\bra{\psi}_{AB} + t \sigma_{AB} \right) \in \PPT \, . \label{eq:PPTLemma}
    \end{equation}
    We follow the proof of \cite{vidal99} and extend it from separable to $\PPT$ states.
    Given a $\sigma$ that fulfills~\cref{eq:PPTLemma}, we show that $t$ has to be larger or equal to $R_{\SEP}(\ket{\psi})$.  Using the property of $\PPT$ states, one obtains for any $\ket{\phi}$
    \begin{align}
        \bra{\phi} \frac{1}{1+t}\left( \ket{\psi}\bra{\psi}_{AB} + t \sigma_{AB} \right)^{T_B} \ket{\phi} \geq 0 \qquad \textnormal{and} \qquad
        \bra{\phi} \left( \ket{\psi}\bra{\psi}_{AB}^{T_B} + t \sigma_{AB}^{T_B} \right) \ket{\phi} \geq 0 \, , \label{eq:PPT_pos}
    \end{align}
    In the following, we want to work in the Schmidt basis $\ket{\phi} = \sum_i a_i \ket{i}_A\ket{i}_B$.
    Using this basis, we now define Bell states
    \begin{align}
        \ket{\phi_{ij}^\perp} = \frac{1}{\sqrt{2}}\left( \ket{ij} + \ket{ji} \right) \qquad \textnormal{and} \qquad
        \ket{\phi_{ij}} = \frac{1}{\sqrt{2}}\left( \ket{ij} - \ket{ji} \right) 
    \end{align}
    and corresponding projectors
    \begin{align}
        P_{ij}^+ = \ket{\phi_{ij}^\perp}\bra{\phi_{ij}^\perp} \qquad \textnormal{and} \qquad
        P_{ij} = \ket{\phi_{ij}}\bra{\phi_{ij}} \, .
    \end{align}
    This allows us to write
    \begin{align}
        \ket{\psi}\bra{\psi}^{T_B} &= \left(\sum_i a_i \ket{ii} \sum_j a_j \bra{jj} \right)^{T_B} \\
        &= \left( \sum_i a_i^2 \ket{ii}\bra{ii} + \sum_{i\neq j}     a_i a_j \ket{ii}\bra{jj} \right)^{T_B} \\
        &= \sum_i a_i^2 \ket{ii}\bra{ii} + \sum_{i\neq j}     a_i a_j \ket{ij}\bra{ji} \\
        &= \sum_i a_i^2 \ket{ii}\bra{ii} + \sum_{i> j}     a_i a_j (\ket{ij}\bra{ji} + \ket{ji}\bra{ij}) \\
        &= \sum_i a_i^2 \ket{ii}\bra{ii} + \sum_{i> j}     a_i a_j ( P_{ij}^+ - P_{ij}) \, .
    \end{align}
    Since~\cref{eq:PPT_pos} holds for any $\ket{\phi}$, we can pick $\ket{\phi} = \ket{\phi_{ij}}$. This leads to
    \begin{equation}
        \bra{\phi_{ij}} \left( -a_i a_j P_{ij} + t \sigma^{T_B} \right) \ket{\phi_{ij}} \geq 0
    \end{equation}
    or equivalently
    \begin{equation}\label{eq:ppt-projector}
        \tr[P_{ij} \sigma^{T_B}] \geq \frac{a_i a_j}{t} \, .
    \end{equation}
    From \cite{vidal99}, we know that the robustness of entanglement can be written as
    \begin{align} \label{eq_vidal}
        R_{\SEP}(\ket{\psi}) &= \Big(\sum_i a_i\Big)^2 -1 
        = \sum_{ij} a_i a_j - 1 
        = \sum_{ij} a_i a_j - \sum_i a_i^2 
        = 2\sum_{i>j} a_i a_j \, .
    \end{align}
    By summing over \cref{eq:ppt-projector}, we then find
    \begin{equation}
        \tr\Big[ \sum_{i>j} P_{ij} \sigma^{T_B} \Big] \geq \frac{R_{\SEP}(\ket{\psi})}{2t} \, .
    \end{equation}
    Thus, if we can show that $\tr[ \sum_{i>j} P_{ij} \sigma^{T_B} ] \leq 1/2$, the statement is shown. Let us define $M \coloneqq \sum_{i>j} P_{ij}$, then using cyclicity of the trace, we have
    \begin{equation}
        \tr[M \sigma^{T_B}] 
        = \tr[M^{T_B} \sigma]
        \leq \lambda_{\max}(M^{T_B}) \, ,
    \end{equation}
    where $\lambda_{\max}(M^{T_B})$ denotes the largest eigenvalue of $M^{T_B}$. To bound this quantity we write 
    \begin{align}
        M^{T_B} &= \left(\sum_{i>j} P_{ij} = \frac{1}{2} \sum_{i>j} \ket{ij}\bra{ij} + \ket{ji}\bra{ji} -\ket{ij}\bra{ji} - \ket{ji}\bra{ij} \right)^{T_B} \\
        &= \frac{1}{2} \sum_{i>j} \ket{ij}\bra{ij} + \ket{ji}\bra{ji} -\ket{ii}\bra{jj} - \ket{jj}\bra{ii} \\
        &= \frac{1}{2} \sum_{i\neq j} \ket{ij}\bra{ij} -\ket{ii}\bra{jj} \, .
    \end{align}
    In this form, it is easy to see that $M^{T_B} \leq \frac{1}{2} \mathbb{1}$, since
    \begin{align}
        \frac{1}{2} \mathbb{1} - M^{T_B} &= \frac{1}{2} \sum_{ij} \ket{ij}\bra{ij} - \frac{1}{2} \sum_{i\neq j} \ket{ij}\bra{ij} + \frac{1}{2} \sum_{i\neq j} \ket{ii}\bra{jj} \\
        &= \frac{1}{2} \left( \sum_i \ket{ii}\ket{ii} +  \sum_{i\neq j} \ket{ii}\bra{jj} \right) \\
        &= \frac{1}{2} (\sum_i \ket{ii}) (\sum_j \bra{jj})
    \end{align}
    which is proportional to a rank 1 projector and therefore greater equal 0. Therefore the largest eigenvalue can be bounded by $1/2$ and we get
    \begin{equation}
        t \geq R_{\SEP}(\ket{\psi})\, ,
    \end{equation}
    which concludes the proof. 
\end{proof}

In previous work (see \cite{piv23,SPS24}), the main technical tool to find lower bounds for $\gamma_{\LO^\star}$ and $\gamma_{\LOCC^\star}$ was not based on $\gamma_{\PPT}$ but on a different approach.
Consider any bipartite channel $\cE\in\CPTP(AB\rightarrow A'B')$ and any separable state $\sigma\in\SEP(A\bar{A},B\bar{B})$ for some ancillary systems $\bar{A},\bar{B}$.
Then one has
\begin{equation} \label{eq_choi_bound_trick}
    \gamma_{\LOCC^{\star}}(\cE) \geq \gamma_{\SEP}\big(\cE(\sigma)\big) \, ,
\end{equation}
since any QPD of $\cE$ directly induces a QPD of $\cE(\sigma)$ with same $\ell_1$-norm of its coefficients.
Typically, $\sigma$ is chosen to be the product of maximally entangled states across $A,\bar{A}$ and $B,\bar B$, respectively.
The resulting state $\cE(\sigma)$ is then the Choi state of $\cE$.


The following example illustrates that $\gamma_{\PPT}(\cU)$ of a unitary channel $\cU$ can be a strictly better lower bound than using a lower bound based on the Choi state.
\begin{example}\label{ex_toffoli}
  Consider the Toffoli gate. There are two ways (up to permutation of control qubits) how to group the three qubits into two bipartitions. By solving the SDP given in~\cref{eq_PPT_SDP} we find for any bipartition that
    \begin{equation}
        \gamma_{\PPT}(\mathrm{Toffoli}) = 3 > \gamma_{\SEP}\big(\cJ(\mathrm{Toffoli})\big)= \frac{1}{4}\Big(\sqrt{6}+\sqrt{2}\Big)^2 - 1 \approx 2.73 \, ,
    \end{equation}
where the penultimate step follows from~\cref{lem_robustness,eq_vidal}.
  It has been shown in~\cite{SPS24} that $\gamma_{\LO^\star}(\mathrm{Toffoli})=3$, therefore the $\PPT$ lower bound turns out to be tight in this case.
\end{example}
As another consequence of~\Cref{prop_ppt_pure_states}, one can also show that the $\PPT$ lower bound is tight for a large class of bipartite unitary channels.
\begin{corollary} \label{cor_PPT_tight}
  Consider a bipartite unitary of the form
  \begin{equation}
    U_{AB} = \left(V_A^{(1)}\otimes V_B^{(2)}\right)  \sum_{k} u_{k} (L_k)_A \otimes (R_k)_B  \left(V_A^{(3)}\otimes V_B^{(4)}\right)  \, ,
  \end{equation}
  where $V_A^{(1)},  V_B^{(2)}, V_A^{(3)}, V_B^{(4)}, (L_k)_A$ and $(R_k)_B$ are unitaries on $A$ and $B$ such that the $(L_k)_k$ and $(R_k)_k$ form orthogonal sets under the Hilbert-Schmidt inner product.
  Denote the channel induced by $U_{AB}$ by $\cU(\rho)\coloneqq U\rho U^{\dagger}$.
  Then
  \begin{align}
    \gamma_{\PPT}(\cU) =
    \gamma_{\LOCC^{\star}}(\cU) =
    \gamma_{\LO^{\star}}(\cU) =
    1 + 2 \sum_{k \ne k'} |u_k| |u_{k'}| \, . \label{eq_ineq}
  \end{align}
\end{corollary}
\begin{proof}
    This is a consequence of~\cite[Theorem 5.1]{SPS24}.
    By simple considerations of hierarchies in the decomposition sets as well as application of the Choi bound and~\cref{prop_ppt_pure_states}, one directly obtains
    \begin{equation}
        \gamma_{\LO^\star}(\cU) \geq \gamma_{\LOCC^\star}(\cU) \geq \gamma_{\PPT^\star}(\cU) \geq \gamma_{\PPT}(\cJ(\cU)) = \gamma_{\SEP}(\cJ(\cU))\, .
    \end{equation}
    The result~\cite[Theorem 5.1]{SPS24} then shows that $\gamma_{\LO^\star}(\cU) = 1 + 2 \sum_{k \ne k'} |u_k| |u_{k'}| =  \gamma_{\SEP}(\cJ(\cU))$
    which concludes the proof.
\end{proof}
Notice, that the unitaries described in~\cref{cor_PPT_tight} contain a large class of unitaries that are used in practice.
Any two-qubit unitary (and tensor products thereof) have the above form due to Cartan's KAK decomposition~\cite{KHANEJA200111}.

An interesting consequence of~\Cref{prop_ppt_pure_states} is that for unitary channels the lower bound $\gamma_{\LOCC^{\star}}\geq \gamma_{\PPT}$ is at least as good as state-based lower bounds, even if one optimizes over the state.
\begin{corollary}\label{corr_statebound}
    For any bipartite unitary channel $\cU\in\CPTP(A B)$ one has
    \begin{equation}
        \gamma_{\LOCC^\star}(\cU) \geq \gamma_{\PPT}(\cU) \geq \sup_{\sigma \in \states(A A' B B')} \frac{\gamma_{\SEP}(\cU(\sigma))}{\gamma_{\SEP}(\sigma)} \, ,
    \end{equation}
    where $A'$ and $B'$ are arbitrarily large ancillary systems.
\end{corollary}

\begin{proof}
The first inequality in the assertion of~\cref{corr_statebound} follows by noting that
\begin{align}
\gamma_{\LOCC^\star}(\cU)  \geq  \gamma_{\PPT^\star}(\cU)   \overset{\textnormal{\Cshref{prop_PPT_PTT_star}}}{=}  \gamma_{\PPT}(\cU) \, .
\end{align}

It thus remains to prove the second inequality of~\cref{corr_statebound}.
    For any choice of $\sigma \in \states(A A' B B')$, we have
    \begin{equation}\label{eq_resource}
        \gamma_{\LOCC^{\star}}(\cU) \gamma_{\LOCC^{\star}}(\sigma) \geq \gamma_{\LOCC^{\star}}( \cU(\sigma)) \, ,
    \end{equation}
    since two QPD for $\cU$ and $\sigma$ can be combined into a QPD for $\cU(\sigma)$ with $\ell_1$-norm being the product of the $\ell_1$-norms of the original QPDs.
    In~\cite[Theorem 1]{Campbell_2010}, it was shown that any logarithmic decomposition-based entanglement monontone attains the worst-case on the set of free states.
    Concretely, this implies
    \begin{equation}\label{resource}
        \sup_{\sigma \in \states(A A' B B') } \frac{\gamma_{\LOCC^\star}(\cU(\sigma))}{\gamma_{\LOCC^\star}(\sigma)} 
        = \sup_{\rho \in \SEP(AA',BB')} \gamma_{\LOCC^\star}(\cU(\rho)) \, .
    \end{equation}
    Additionally, we have $\gamma_{\LOCC^\star}(\rho) = \gamma_{\SEP}(\rho)$ for states, as explained in~\cref{lem_extent_states}.
    We then find
    \begin{align}
        \sup_{\sigma \in \states(A A'B B') } \frac{\gamma_{\LOCC^\star}(\cU(\sigma))}{\gamma_{\LOCC^\star}(\sigma)} 
        \overset{\textnormal{\Cshref{resource}}}&{=} \sup_{\rho \in \SEP(AA',BB')} \gamma_{\LOCC^\star}(\cU(\rho)) \\
        \overset{\textnormal{\Cshref{lem_extent_states}}}&{=} \sup_{\rho \in \SEP(AA',BB')} \gamma_{\SEP}(\cU(\rho)) \\
        &= \sup_{\psi\in\SEP(AA',BB')} \gamma_{\SEP}(\cU(\ket{\psi}\bra{\psi})) \, ,
\end{align}
where the final step uses that $\gamma_{\SEP}$ is a convex quantity, and optimization of it over all states reduces to an optimization over all pure states.
Using that $\gamma_{\SEP}(\ket{\psi}\bra{\psi}) =1$ for any $\psi\in\SEP$ yields
\begin{align}  
\sup_{\psi\in\SEP(AA',BB')} \gamma_{\SEP}(\cU(\ket{\psi}\bra{\psi}))
         &= \sup_{\psi\in\SEP(AA',BB')} \frac{\gamma_{\SEP}(\cU(\ket{\psi}\bra{\psi}))}{\gamma_{\SEP}(\ket{\psi}\bra{\psi})} \\       
        &\leq \sup_{\psi \in \overline{\mathrm{D}}(AA'BB')} \frac{\gamma_{\SEP}(\cU(\ket{\psi}\bra{\psi}))}{\gamma_{\SEP}(\ket{\psi}\bra{\psi})} \\
        \overset{\textnormal{\Cshref{prop_ppt_pure_states}}}&{=} \sup_{\psi \in \overline{\mathrm{D}}(AA'BB')} \frac{\gamma_{\PPT}(\cU(\ket{\psi}\bra{\psi}))}{\gamma_{\PPT}(\ket{\psi}\bra{\psi})} \\
        &\leq \sup_{\sigma \in \states(AA'BB')} \frac{\gamma_{\PPT}(\cU(\sigma))}{\gamma_{\PPT}(\sigma)} \\
        \overset{\textnormal{\Cshref{eq_resource}}}&{\leq} \gamma_{\PPT}(\cU)   \, .     
    \end{align}
\end{proof}

\section{Discussion}
In this work, we made the interesting observation that classical side information is absolutely crucial to realize circuit cutting without classical communication, whereas it provides no advantage in the $\SEP$ and $\PPT$ setting.
This raises the question: does classical side information help to reduce the quasiprobability extent in the case of $\LOCC$?
In mathematical terms, do there exist channels $\cE$ such that $\gamma_{\LOCC^\star}(\cE)<\gamma_{\LOCC}(\cE)$?
As noted in the discussion of \Cref{lemma_s_s_star}, our current techniques do not suffice to answer this question.
Due to the notoriously complicated structure of $\LOCC$, this problem could be challenging to resolve.

We do note that for some nontrivial families of channels, it is known that $\gamma_{\LOCC^\star}=\gamma_{\LOCC}$.
For instance, it has been shown in~\cite{piv23} that this is the case for Clifford unitaries.
On the other hand, the optimal $\LOCC^{\star}$ QPDs for general two-qubit unitaries introduced in~\cite{SPS24} heavily utilize the classical side information, and it seems unclear how the optimal overhead could also be achieved without it.
This could be seen as evidence for a separation between $\gamma_{\LOCC^{\star}}$ and $\gamma_{\LOCC}$.

More broadly, it remains an important task for future work to better understand the quasiprobability extent of unitary channels.
For unitaries, the quasiprobability extent $\gamma_{\LOCC^{\star}}(\cU)$ is currently only known in two cases: for Clifford gates~\cite{piv23} and for unitaries with Schmidt-like decompositions (including two-qubit unitaries)~\cite{SPS24,anguspaper}.
In both cases, our work shows that the PPT lower bound is tight, i.e.~$\gamma_{\LOCC^{\star}}(\cU)=\gamma_{\PPT}(\cU)$.\footnote{For Clifford gates, the statement follows from \Cref{prop_ppt_pure_states}, \Cref{corr_statebound}, and $\gamma_{\LOCC^{\star}}(\cU)=\gamma_{\SEP}(\cJ(\cU))$, which is proven in~\cite[Theorem 5.1]{piv23}.}
It is natural to ask whether this equality also holds for general bipartite unitaries.
Such a result would lift the statement of \Cref{prop_ppt_pure_states} from the setting of states to channels.
From a more practical side, this would allow us to efficiently evaluate the quasiprobability extent for arbitrary unitary channels, which is relevant for circuit cutting applications.


\paragraph{Acknowledgements}
CP and LS acknowledge support from the ETH Quantum Center and the NCCR SwissMAP.

\appendix

\section{Extent separations}\label{app:remaining_inequalities}
In this section, we show that the quasiprobability extents of different decomposition sets differ as depicted by the horizontal inequalities in~\Cref{fig_results}.
In technical term we show that
\begin{equation}\label{eq:horizontal_ineq_1}
    \gamma_{\LO^\star} > \gamma_{\LOCC^\star} > \gamma_{\SEP^\star} > \gamma_{\PPT^\star}
\end{equation}
and
\begin{equation}\label{eq:horizontal_ineq_2}
    \gamma_{\LO} > \gamma_{\LOCC} > \gamma_{\SEP} > \gamma_{\PPT}\, ,
\end{equation}
where $\gamma_{S_1} > \gamma_{S_2}$ means that there exits a channel $\cE$ such that $\gamma_{S_1}(\cE) > \gamma_{S_2}(\cE)$.
This is a direct consequence of the following result.
\begin{lemma}
    Consider two sets of quantum instruments $Q_1\subset\qi(I \rightarrow O)$, $Q_2\subset\qi(I \rightarrow O)$ and denote the induced sets of channels by $S_i \coloneqq \{ \sum_j\cE_j : (\cE_j)_j \in Q_i \}$ and $S_i^{\star}\coloneqq\cL[Q_i]$  for $i=1,2$.
    Furthermore, assume that \smash{$\overline{\conv(S_1)} \subsetneqq S_2$} where \smash{$\overline{\conv(S_1)}$} denotes the closure of the convex hull of $S_1$.
    Then, there exists a channel $\cE\in\cptp(I \rightarrow O)$ such that
    \begin{equation}
        \gamma_{S_2}(\cE) = \gamma_{S_2^{\star}}(\cE) = 1 \qquad\text{and}\qquad \gamma_{S_1}(\cE),\gamma_{S_1^{\star}}(\cE) > 1 \, .
    \end{equation}
\end{lemma}
Our desired statement in \Cref{eq:horizontal_ineq_1,eq:horizontal_ineq_2} is a direct consequence of the above lemma, by appropriately choosing $Q_1,Q_2$ to be either $\qi_{\LO},\qi_{\LOCC}$, $\qi_{\LOCC},\qi_{\SEP}$ or $\qi_{\SEP},\qi_{\PPT}$.
Indeed, it is known that $\overline{\conv(\LO)}$ (which contains local operations with shared randomness) is a strict subset of $\LOCC$, the closure $\overline{\LOCC}$ is a strict subset of $\SEP$ (this has been demonstrated in the context of state discrimination~\cite{bennett1999_quantum,koashi2007_quantum,childs2013_framework} and random distillation~\cite{cui2011_randomly,chitambar2012_increasing}) and that $\SEP$ (which is convex and closed) is a strict subset of $\PPT$ (see e.g.~\cite{bennett1999_unextendible}).

\begin{proof}
    Let us pick some $\cE$ that lies in $S_2$ but not in $\overline{\conv(S_1)}$.
    Clearly, $\gamma_{S_2}(\cE) = \gamma_{S_2^{\star}}(\cE) = 1$.
    We will now prove by contradiction that $\gamma_{S_1^{\star}}(\cE)$ must be $>1$.
    This in turn will also imply that $\gamma_{S_1}(\cE)>1$, since $S_1\subset S_1^{\star}$.

    Let us assume that $\gamma_{S_1^{\star}}(\cE)=1$.
    Therefore, for any $\epsilon >0$ there exists a QPD $\cE=\sum_i a_i\cF_i$ with $\cF_i\in S_1^{\star}$ and $\norm{a}_1\leq 1+\epsilon$.
    We can write $\cF_i=\sum_j b_j^{(i)}\cG^{(i)}_j$ with $b_j^{(i)} \in [-1,1]$ and $(\cG_j^{(i)})_j\in Q_1$.

    Now, define the superoperator $\tilde{\cE}\coloneqq \sum_{i,j}\abs{a_i}\cG_{j}^{(i)}$.
    For the Choi states of $\cE$ and $\tilde{\cE}$ it holds that
    \begin{equation}
        \cJ(\cE) = \sum_{i,j} a_i b_j^{(i)} \cJ(G_j^{(i)}) \leq \sum_{i,j} \abs{a_i} \cJ(G_j^{(i)}) = \cJ(\tilde{\cE}) \, .
    \end{equation}
    Hence, the trace distance between these two Choi states is
    \begin{align}
        \norm{\cJ(\cE)-\cJ(\tilde{\cE})}_1 = \tr[\cJ(\tilde{\cE})-\cJ(\cE)] 
        = \sum_{i,j} |a_i| \tr[\cJ(G_j^{(i)})] -1
        = \sum_i |a_i| - 1
        \leq \epsilon \, ,
    \end{align}
    The channel $\frac{1}{\norm{a}_1}\tilde{\cE}$ lies in $\conv(S_1)$ and approximates $\cE$ up to error
    \begin{equation}
        \norm{\cJ(\cE)-\frac{1}{\norm{a}_1}\cJ(\tilde{\cE})}_1 \leq \norm{\cJ(\cE)-\cJ(\tilde{\cE})}_1 + \norm{\cJ(\tilde{\cE})-\frac{1}{\norm{a}_1}\cJ(\tilde{\cE})}_1 \leq \epsilon + \frac{\epsilon}{1+\epsilon} \leq 2\epsilon \, .
    \end{equation}
    In summary, we have shown that our map $\cE$ can be approximated arbitrarily well by elements in $\conv(S_1)$.
    This implies that $\cE$ lies in the closure of $\conv(S_1)$, which is a contradiction.
\end{proof}

\bibliographystyle{arxiv_no_month}
\bibliography{bibliofile}

\end{document}

%% file: results_fig.tex
\begin{tikzpicture}
\def \x{2};
\def \y{1.3};
\draw[gray!50,fill=gray!50] (0,0) ellipse (0.6 and 0.4);
 \node at (0,0) {$\gamma_{\mathrm{LO}^{\star}}$}; 
 \node at (0,-\y) {$\gamma_{\mathrm{LO}}$}; 
 \draw[gray!50,fill=gray!50] (\x,0) ellipse (0.6 and 0.4);
 \node at (\x,0) {$\gamma_{\mathrm{LOCC}^{\star}}$}; 
 \node at (\x,-\y) {$\gamma_{\mathrm{LOCC}}$}; 
  \node at (2*\x,0) {$\gamma_{\mathrm{SEP}^{\star}}$}; 
 \node at (2*\x,-\y) {$\gamma_{\mathrm{SEP}}$}; 
 \node at (3*\x,0) {$\gamma_{\mathrm{PPT}^{\star}}$}; 
 \draw[ForestGreen!20,fill=ForestGreen!20] (3*\x,-\y) ellipse (0.6 and 0.4); 
 \node at (3*\x,-\y) {$\gamma_{\mathrm{PPT}}$};  

\node[gray] at (-1.3*\x,0) {operationally relevant};
\node[ForestGreen] at (4.5*\x,-\y) {efficiently computable (SDP)};

 \node[rotate=90] at (0,-0.5*\y) {\color{linkblue}$>$};
  \node[rotate=90] at (\x,-0.5*\y) {$\geq$};
    \node[rotate=90] at (2*\x,-0.5*\y) {\color{linkblue}$=$};
  \node[rotate=90] at (3*\x,-0.5*\y) {\color{linkblue}$=$};
  \node at (0.5*\x,0) {$>$}; 
  \node at (1.5*\x,-\y) {$>$}; 
  \node at (0.5*\x,-\y) {$>$}; 
  \node at (1.5*\x,0) {$>$};
  \node at (2.5*\x,0) {$>$};  
  \node at (2.5*\x,-\y) {$>$};

  \node at (-1.1,-0.5*\y) {\Cshref{prop_LO_vs_LO_star}};
  \node at (1.8*\x+1.4,-0.5*\y) {\Cshref{prop_PPT_PTT_star}};

\end{tikzpicture}

%% file: figure_cutting_example.tex
    \begin{tikzpicture}[thick, scale=1]
    \def\xs{0.07}
    \def \b{0.3}
    \def \x{8}
    \def \s{0.02}

     \draw [fill=gray!15,draw=none] (-0.05,0.75+2.5*\xs) rectangle (5.4,0+0.25*\xs);   
     \draw [fill=cyan!15,draw=none] (-0.05,-0.25*\xs) rectangle (5.4,-0.75-2.5*\xs);       
    
    \draw (-0.1,0.45) -- (5,0.45);   
    \draw (-0.1,0.75) -- (5,0.75);  
    \draw (-0.1,0.15) -- (5,0.15);
    \draw (-0.1,-0.15) -- (5,-0.15);
    \draw (-0.1,-0.45) -- (5,-0.45);   
    \draw (-0.1,-0.75) -- (5,-0.75);   
    \draw [fill=yellow!40,draw=black] (1,0.15+\xs) rectangle (1+\b,-0.15-\xs); 
    \draw [fill=yellow!40,draw=black] (3,0.15+\xs) rectangle (3+\b,-0.45-\xs); 
    \draw [fill=yellow!40,draw=black] (4,0.75+\xs) rectangle (4+\b,-0.15-\xs); 
    \node at (1+\b/2+0.01,0) {\footnotesize{$U_{\!1}$}};
    \node at (3+\b/2+0.01,-0.15) {\footnotesize{$U_{\!2}$}};
    \node at (4+\b/2+0.01,0.3) {\footnotesize{$U_{\!3}$}};    
    
    \draw [fill=white,draw=black] (0.5,0.75+\xs) rectangle (0.5+\b,0.45-\xs);
    \draw [fill=white,draw=black] (0.0,0.15+\xs) rectangle (0.0+\b,0.15-\xs);   
    \draw [fill=white,draw=black] (0.0,-0.15+\xs) rectangle (0.0+\b,-0.45-\xs);
    \draw [fill=white,draw=black] (0.5,-0.45+\xs) rectangle (0.5+\b,-0.75-\xs);
    
    \draw [fill=white,draw=black] (1,0.75+\xs) rectangle (1+\b,0.75-\xs);
    \draw [fill=white,draw=black] (1.5,0.45+\xs) rectangle (1.5+\b,0.15-\xs);  
    \draw [fill=white,draw=black] (1.5,-0.15+\xs) rectangle (1.5+\b,-0.45-\xs);
    \draw [fill=white,draw=black] (2,-0.75+\xs) rectangle (2+\b,-0.75-\xs);
    \draw [fill=white,draw=black] (2,0.75+\xs) rectangle (2+\b,0.15-\xs);

    \draw [fill=white,draw=black] (2.5,-0.45+\xs) rectangle (2.5+\b,-0.75-\xs);  
    
    \draw [fill=white,draw=black] (3,0.75+\xs) rectangle (3+\b,0.45-\xs); 
    
    \draw [fill=white,draw=black] (3.5,0.45+\xs) rectangle (3.5+\b,0.15-\xs);
    \draw [fill=white,draw=black] (3.5,-0.45+\xs) rectangle (3.5+\b,-0.75-\xs);

    \draw [fill=white,draw=black] (4.5,-0.15+\xs) rectangle (4.5+\b,-0.45-\xs);
    \draw [fill=white,draw=black] (4.5,0.75+\xs) rectangle (4.5+\b,0.45-\xs);  
    
    \draw [fill=white,draw=black] (5,0.15+1.6*\xs) rectangle (5+\b,0.15-1.6*\xs);
    \draw [fill=white,draw=black] (5,0.45+1.6*\xs) rectangle (5+\b,0.45-1.6*\xs); 
    \draw [fill=white,draw=black] (5,0.75+1.6*\xs) rectangle (5+\b,0.75-1.6*\xs); 
    \draw [fill=white,draw=black] (5,-0.15+1.6*\xs) rectangle (5+\b,-0.15-1.6*\xs);
    \draw [fill=white,draw=black] (5,-0.45+1.6*\xs) rectangle (5+\b,-0.45-1.6*\xs); 
    \draw [fill=white,draw=black] (5,-0.75+1.6*\xs) rectangle (5+\b,-0.75-1.6*\xs);               
     \draw[thin] (5+\b-0.05,0.15-0.07) arc (0:180:0.1);
     \draw[thin,->] (5+\b-0.05-0.11,0.15-0.07) -- (5+\b-0.05,0.15-0.07+0.15); 
     \draw[thin] (5+\b-0.05,0.45-0.07) arc (0:180:0.1);
     \draw[thin,->] (5+\b-0.05-0.11,0.45-0.07) -- (5+\b-0.05,0.45-0.07+0.15);   
     \draw[thin] (5+\b-0.05,0.75-0.07) arc (0:180:0.1);
     \draw[thin,->] (5+\b-0.05-0.11,0.75-0.07) -- (5+\b-0.05,0.75-0.07+0.15);     
     \draw[thin] (5+\b,0.15+\s) -- (5+\b+7*\s,0.15+\s);
     \draw[thin] (5+\b,0.15-\s) -- (5+\b+7*\s,0.15-\s);     
     \draw[thin] (5+\b,0.45+\s) -- (5+\b+7*\s,0.45+\s);
     \draw[thin] (5+\b,0.45-\s) -- (5+\b+7*\s,0.45-\s);   
     \draw[thin] (5+\b,0.75+\s) -- (5+\b+7*\s,0.75+\s);
     \draw[thin] (5+\b,0.75-\s) -- (5+\b+7*\s,0.75-\s);   

     \draw[thin] (5+\b-0.05,-0.15-0.07) arc (0:180:0.1);
     \draw[thin,->] (5+\b-0.05-0.11,-0.15-0.07) -- (5+\b-0.05,-0.15-0.07+0.15); 
     \draw[thin] (5+\b-0.05,-0.45-0.07) arc (0:180:0.1);
     \draw[thin,->] (5+\b-0.05-0.11,-0.45-0.07) -- (5+\b-0.05,-0.45-0.07+0.15);   
     \draw[thin] (5+\b-0.05,-0.75-0.07) arc (0:180:0.1);
     \draw[thin,->] (5+\b-0.05-0.11,-0.75-0.07) -- (5+\b-0.05,-0.75-0.07+0.15);     
     \draw[thin] (5+\b,-0.15+\s) -- (5+\b+7*\s,-0.15+\s);
     \draw[thin] (5+\b,-0.15-\s) -- (5+\b+7*\s,-0.15-\s);     
     \draw[thin] (5+\b,-0.45+\s) -- (5+\b+7*\s,-0.45+\s);
     \draw[thin] (5+\b,-0.45-\s) -- (5+\b+7*\s,-0.45-\s);   
     \draw[thin] (5+\b,-0.75+\s) -- (5+\b+7*\s,-0.75+\s);
     \draw[thin] (5+\b,-0.75-\s) -- (5+\b+7*\s,-0.75-\s);       
                          
\node at (5/2+\b/2-0.2/2+\x/2,0) {$\implies$};
     \draw [fill=gray!15,draw=none] (-0.05+\x,0.75+2.5*\xs) rectangle (5.4+\x,0+0.25*\xs);   
     \draw [fill=cyan!15,draw=none] (-0.05+\x,-0.25*\xs) rectangle (5.4+\x,-0.75-2.5*\xs);       
    
    \draw (-0.1+\x,0.45) -- (5+\x,0.45);   
    \draw (-0.1+\x,0.75) -- (5+\x,0.75);  
    \draw (-0.1+\x,0.15) -- (5+\x,0.15);
    \draw (-0.1+\x,-0.15) -- (5+\x,-0.15);
    \draw (-0.1+\x,-0.45) -- (5+\x,-0.45);   
    \draw (-0.1+\x,-0.75) -- (5+\x,-0.75);   
    \draw [fill=yellow!40,draw=black] (1+\x,0.15+\xs) rectangle (1+\x+\b,0.15-\xs); 
    \draw [fill=yellow!40,draw=black] (1+\x,-0.15+\xs) rectangle (1+\x+\b,-0.15-\xs);     
    \draw [fill=yellow!40,draw=black] (3+\x,0.15+\xs) rectangle (3+\x+\b,0.15-\xs); 
    \draw [fill=yellow!40,draw=black] (3+\x,-0.15+\xs) rectangle (3+\x+\b,-0.45-\xs);     
    \draw [fill=yellow!40,draw=black] (4+\x,0.75+\xs) rectangle (4+\x+\b,0.15-\xs); 
    \draw [fill=yellow!40,draw=black] (4+\x,-0.15+\xs) rectangle (4+\x+\b,-0.15-\xs);     
    
    \draw [fill=white,draw=black] (0.5+\x,0.75+\xs) rectangle (0.5+\b+\x,0.45-\xs);
    \draw [fill=white,draw=black] (0.0+\x,0.15+\xs) rectangle (0.0+\b+\x,0.15-\xs);   
    \draw [fill=white,draw=black] (0.0+\x,-0.15+\xs) rectangle (0.0+\b+\x,-0.45-\xs);
    \draw [fill=white,draw=black] (0.5+\x,-0.45+\xs) rectangle (0.5+\b+\x,-0.75-\xs);
    
    \draw [fill=white,draw=black] (1+\x,0.75+\xs) rectangle (1+\b+\x,0.75-\xs);
    \draw [fill=white,draw=black] (1.5+\x,0.45+\xs) rectangle (1.5+\b+\x,0.15-\xs);  
    \draw [fill=white,draw=black] (1.5+\x,-0.15+\xs) rectangle (1.5+\b+\x,-0.45-\xs);
    \draw [fill=white,draw=black] (2+\x,-0.75+\xs) rectangle (2+\b+\x,-0.75-\xs);
    \draw [fill=white,draw=black] (2+\x,0.75+\xs) rectangle (2+\b+\x,0.15-\xs);

    \draw [fill=white,draw=black] (2.5+\x,-0.45+\xs) rectangle (2.5+\b+\x,-0.75-\xs);  
    
    \draw [fill=white,draw=black] (3+\x,0.75+\xs) rectangle (3+\b+\x,0.45-\xs); 
    
    \draw [fill=white,draw=black] (3.5+\x,0.45+\xs) rectangle (3.5+\b+\x,0.15-\xs);
    \draw [fill=white,draw=black] (3.5+\x,-0.45+\xs) rectangle (3.5+\b+\x,-0.75-\xs);

    \draw [fill=white,draw=black] (4.5+\x,-0.15+\xs) rectangle (4.5+\b+\x,-0.45-\xs);
    \draw [fill=white,draw=black] (4.5+\x,0.75+\xs) rectangle (4.5+\b+\x,0.45-\xs);  
    
    \draw [fill=white,draw=black] (5+\x,0.15+1.6*\xs) rectangle (\x+5+\b,0.15-1.6*\xs);
    \draw [fill=white,draw=black] (\x+5,0.45+1.6*\xs) rectangle (\x+5+\b,0.45-1.6*\xs); 
    \draw [fill=white,draw=black] (\x+5,0.75+1.6*\xs) rectangle (\x+5+\b,0.75-1.6*\xs); 
    \draw [fill=white,draw=black] (\x+5,-0.15+1.6*\xs) rectangle (\x+5+\b,-0.15-1.6*\xs);
    \draw [fill=white,draw=black] (\x+5,-0.45+1.6*\xs) rectangle (\x+5+\b,-0.45-1.6*\xs); 
    \draw [fill=white,draw=black] (\x+5,-0.75+1.6*\xs) rectangle (\x+5+\b,-0.75-1.6*\xs);               
     \draw[thin] (\x+5+\b-0.05,0.15-0.07) arc (0:180:0.1);
     \draw[thin,->] (\x+5+\b-0.05-0.11,0.15-0.07) -- (\x+5+\b-0.05,0.15-0.07+0.15); 
     \draw[thin] (\x+5+\b-0.05,0.45-0.07) arc (0:180:0.1);
     \draw[thin,->] (\x+5+\b-0.05-0.11,0.45-0.07) -- (\x+5+\b-0.05,0.45-0.07+0.15);   
     \draw[thin] (\x+5+\b-0.05,0.75-0.07) arc (0:180:0.1);
     \draw[thin,->] (\x+5+\b-0.05-0.11,0.75-0.07) -- (\x+5+\b-0.05,0.75-0.07+0.15);     
     \draw[thin] (\x+5+\b,0.15+\s) -- (\x+5+\b+7*\s,0.15+\s);
     \draw[thin] (\x+5+\b,0.15-\s) -- (\x+5+\b+7*\s,0.15-\s);     
     \draw[thin] (\x+5+\b,0.45+\s) -- (\x+5+\b+7*\s,0.45+\s);
     \draw[thin] (\x+5+\b,0.45-\s) -- (\x+5+\b+7*\s,0.45-\s);   
     \draw[thin] (\x+5+\b,0.75+\s) -- (\x+5+\b+7*\s,0.75+\s);
     \draw[thin] (\x+5+\b,0.75-\s) -- (\x+5+\b+7*\s,0.75-\s);   

     \draw[thin] (\x+5+\b-0.05,-0.15-0.07) arc (0:180:0.1);
     \draw[thin,->] (\x+5+\b-0.05-0.11,-0.15-0.07) -- (\x+5+\b-0.05,-0.15-0.07+0.15); 
     \draw[thin] (\x+5+\b-0.05,-0.45-0.07) arc (0:180:0.1);
     \draw[thin,->] (\x+5+\b-0.05-0.11,-0.45-0.07) -- (\x+5+\b-0.05,-0.45-0.07+0.15);   
     \draw[thin] (\x+5+\b-0.05,-0.75-0.07) arc (0:180:0.1);
     \draw[thin,->] (\x+5+\b-0.05-0.11,-0.75-0.07) -- (\x+5+\b-0.05,-0.75-0.07+0.15);     
     \draw[thin] (\x+5+\b,-0.15+\s) -- (\x+5+\b+7*\s,-0.15+\s);
     \draw[thin] (\x+5+\b,-0.15-\s) -- (\x+5+\b+7*\s,-0.15-\s);     
     \draw[thin] (\x+5+\b,-0.45+\s) -- (\x+5+\b+7*\s,-0.45+\s);
     \draw[thin] (\x+5+\b,-0.45-\s) -- (\x+5+\b+7*\s,-0.45-\s);   
     \draw[thin] (\x+5+\b,-0.75+\s) -- (\x+5+\b+7*\s,-0.75+\s);
     \draw[thin] (\x+5+\b,-0.75-\s) -- (\x+5+\b+7*\s,-0.75-\s);

     \draw [fill=gray!15,draw=none] (0.1,1+0.2) rectangle (0.1+0.25,1+0.25+0.2);
     \draw [fill=cyan!15,draw=none] (1.6,1+0.2) rectangle (1.6+0.25,1+0.25+0.2);     
     \node[gray] at (0.55,1.325) {\footnotesize{$A$}};
     \node[cyan] at (2.05,1.325) {\footnotesize{$B$}};
   
    \end{tikzpicture}

%% file: simple_gamma_fig.tex
\begin{tikzpicture}
\def \x{2};
\def \y{1.3};
 \node at (0,0) {$\gamma_{\mathrm{LO}^{\star}}$}; 
 \node at (0,-\y) {$\gamma_{\mathrm{LO}}$}; 
 \node at (\x,0) {$\gamma_{\mathrm{LOCC}^{\star}}$}; 
 \node at (\x,-\y) {$\gamma_{\mathrm{LOCC}}$}; 
  \node at (2*\x,0) {$\gamma_{\mathrm{SEP}^{\star}}$}; 
 \node at (2*\x,-\y) {$\gamma_{\mathrm{SEP}}$}; 
 \node at (3*\x,0) {$\gamma_{\mathrm{PPT}^{\star}}$}; 
 \node at (3*\x,-\y) {$\gamma_{\mathrm{PPT}}$};  


 \node[rotate=90] at (0,-0.5*\y) {$\geq$};
  \node[rotate=90] at (\x,-0.5*\y) {$\geq$};
    \node[rotate=90] at (2*\x,-0.5*\y) {$\geq$};
  \node[rotate=90] at (3*\x,-0.5*\y) {$\geq$};
  \node at (0.5*\x,0) {$\geq$}; 
  \node at (1.5*\x,-\y) {$\geq$}; 
  \node at (0.5*\x,-\y) {$\geq$}; 
  \node at (1.5*\x,0) {$\geq$};
  \node at (2.5*\x,0) {$\geq$};  
  \node at (2.5*\x,-\y) {$\geq$};


\end{tikzpicture}